%% file: main.tex
\documentclass[leqno,centertags,reqno]{LMCS}

\def\dOi{10(4:21)2014}
\lmcsheading%
{\dOi}
{1--21}
{}
{}
{Jan.~19, 2013}
{Dec.~31, 2014}
{}

\ACMCCS{[{\bf Theory of computation}]: Semantics and
  reasoning---Program reasoning---Program verification; [{\bf
      Mathematics of computing}]: Probability and statistics}

\usepackage{amsmath,amssymb,amsfonts,latexsym,stmaryrd}
\usepackage{pslatex}
\usepackage{epsfig,psfrag}

\usepackage{hyperref}

\usepackage{gastex,delarray,wrapfig}
\usepackage{extarrows}
\usepackage{stmaryrd}
\usepackage{verbatim}
\usepackage[mathscr]{euscript}
\usepackage{color}
\usepackage{pifont}
\usepackage{subfig}
\usepackage{multimedia}
\usepackage{delarray}
\usepackage{graphicx}
\usepackage{tikz}
\usepackage{listings}
\usepackage{float}
\usepackage{pgf}
\usepackage{tikz}
\usetikzlibrary{automata}
\usetikzlibrary{trees}
\usetikzlibrary{shapes}
\usetikzlibrary{shapes.symbols}
\usetikzlibrary{petri}
\usetikzlibrary{arrows}
\usetikzlibrary{snakes}
\usetikzlibrary{backgrounds}
\usetikzlibrary{calc}
\usetikzlibrary{patterns}
\usetikzlibrary{spy}
\tikzstyle{prob-node}=[draw=black,diamond]
\tikzstyle{player-node}=[draw=black,circle]
\tikzstyle{opponent-node}=[draw=black]
\tikzstyle{transition-edge}=[->,line width=1pt]
\tikzset{background rectangle/.style={}}

\newcommand{\hide}[1]{}

 \newcommand{\longleadsto}
    {
      \tikz \draw [->,
      line join=round,
      decorate, decoration={
        zigzag,
        segment length=4,
        amplitude=.9,post=lineto,
        post length=2pt
    }]  (0,0) -- (0.5,0);
    }

\newcommand{\lrb}[1]{[#1]}
\newcommand{\lrc}[1]{(#1)}
\newcommand{\lrd}[1]{\{#1\}}

\newcommand{\ignore}[1]{}

\newcommand{\emptyword}{\varepsilon}

\newcommand{\wpp}{w.p.p.}
\newcommand{\as}{a.s.}

\newcommand{\nat}{\mathbb N}
\newcommand{\ord}{\mathbb O}

\newcommand{\set}[1]{\lrd{#1}}
\newcommand{\setcomp}[2]{\lrd{{#1}:{#2}}}
\newcommand{\tuple}[1]{\lrc{#1}}
\newcommand{\denotationof}[2]{\llbracket #1\rrbracket^{#2}}
\newcommand{\game}{{\mathcal G}}
\newcommand{\tsys}{{\mathcal T}}
\newcommand{\bscc}{B}
\newcommand{\gametuple}{\tuple{\states,\zstates,\ostates,\rstates,\transition,\probp,\coloring}}
\newcommand{\selectfrom}[1]{{\it select}\lrc{#1}}
\newcommand{\restrict}{|}
\newcommand{\undef}{\bot}
\newcommand{\domof}[1]{{\it dom}\lrc{#1}}
\newcommand{\states}{S}
\newcommand{\memstates}{\states_{\memory{}}}
\newcommand{\stateset}{Q}
\newcommand{\targetset}{{\tt Target}}
\newcommand{\state}{s}

\newcommand{\zstates}{\states^0}
\newcommand{\ostates}{\states^1}
\newcommand{\rstates}{\states^R}
\newcommand{\xstates}{\states^x}
\newcommand{\ystates}{\states^{1-x}}
\newcommand{\transition}{{\longrightarrow}}

\newcommand{\tmovesto}{\longleadsto}
\newcommand{\probp}{P}

\newcommand{\parg}[1]{\paragraph{\bf #1}\vspace{-2mm}}
\newcommand{\postof}[2]{{\it Post}_{#1}\left(#2\right)}
\newcommand{\preof}[2]{{\it Pre}_{#1}\left(#2\right)}
\newcommand{\preRof}[2]{{\it Pre}_{#1}^R\left(#2\right)}
\newcommand{\dualpreof}[2]{\widetilde{{\it Pre}_{#1}}\left(#2\right)}

\newcommand{\complementof}[1]{\overline{#1}}
\newcommand{\gcomplementof}[2]{\stackrel{_{#1}}{_\neg}{#2}}
\newcommand{\zcut}[1]{\lrb{#1}^0}
\newcommand{\ocut}[1]{\lrb{#1}^1}
\newcommand{\zocut}[1]{\lrb{#1}^{0,1}}
\newcommand{\rcut}[1]{\lrb{#1}^R}
\newcommand{\xcut}[1]{\lrb{#1}^{x}}
\newcommand{\ycut}[1]{\lrb{#1}^{1-x}}

\newcommand{\run}{\rho}
\newcommand{\runset}{{\mathfrak R}}
\newcommand{\runsof}[1]{{\it Runs}\lrc{#1}}
\newcommand{\pth}{\pi}
\newcommand{\pthset}[1]{{\mathrm\Pi}_{#1}}
\newcommand{\strat}{f}
\newcommand{\zstrat}{\strat^0}
\newcommand{\ostrat}{\strat^1}
\newcommand{\xstrat}{\strat^x}
\newcommand{\ystrat}{\strat^{1-x}}
\newcommand{\xforcestrat}{{\it force}^x}
\newcommand{\yforcestrat}{{\it force}^{1-x}}

\newcommand{\yavoidstrat}{{\it avoid}^{1-x}}

\newcommand{\xallstrats}{F_{\it all}^\px}
\newcommand{\xfinitestrats}{F^\px_{\it finite}}
\newcommand{\xnomemstrats}{F^{\px}_\emptyset}
\newcommand{\yallstrats}{F_{\it all}^{\py}}
\newcommand{\yfinitestrats}{F^{\py}_{\it finite}}
\newcommand{\ynomemstrats}{F^{\py}_\emptyset}

\newcommand{\xstratset}{F^\px}
\newcommand{\ystratset}{F^{\py}}
\newcommand{\zstratset}{F^0}
\newcommand{\ostratset}{F^1}
\newcommand{\px}{x}
\newcommand{\py}{{1-x}}
\newcommand{\pz}{0}
\newcommand{\po}{1}
\newcommand{\partialto}{\rightharpoonup}

\newcommand{\maxcolor}{{n_{\text{max}}}}

\newcommand{\memory}{M}
\newcommand{\memconf}{m}
\newcommand{\memtrans}{\tau}
\newcommand{\memmem}{\mu}
\newcommand{\memstrattuple}{\tuple{\memory,\initmem,\memtrans,\memmem}}
\newcommand{\initmem}{\memconf_0}
\newcommand{\memstrat}[1]{\mathcal M^{#1}}
\newcommand{\memstratn}{\memstrat{}}

\newcommand{\memstratstrat}[1]{\mathit{strat}_{\memstrat{#1}}}
\newcommand{\memstratstratn}{\memstratstrat{}}

\newcommand{\om}{\omega}
\newcommand{\Om}{\Omega}
\newcommand{\probm}{{\mathcal P}}

\newcommand{\formula}{{\varphi}}
\newcommand{\parity}{{\it Parity}}
\newcommand{\xparity}{\mbox{$x$-}\parity}
\newcommand{\yparity}{\mbox{$(1-x)$-}\parity}

\newcommand{\always}{\Box}
\newcommand{\eventually}{\Diamond}
\newcommand{\colorset}[3]{[#1]^{\coloring#2#3}}
\newcommand{\coloring}{{\mathtt{Col}}}
\newcommand{\lcscoloring}{{\mathtt{Col}}}
\newcommand{\colorof}[1]{\coloring\lrc{#1}}

\newcommand{\winset}{W}

\newcommand{\xwinset}{\winset^x}
\newcommand{\ywinset}{\winset^{1-x}}

\newcommand{\vinset}{V}

\newcommand{\xvinset}{\vinset^x}
\newcommand{\yvinset}{\vinset^{1-x}}

\newcommand{\reachset}{{\mathcal R\,\,\!\!}}

\newcommand{\xforceset}{{\it Force}^x}
\newcommand{\yforceset}{{\it Force}^{1-x}}

\newcommand{\yavoidset}{{\it Avoid}^{1-x}}
\newcommand{\tp}{\gamma}

\newcommand{\cset}{{\mathcal C}}
\newcommand{\xset}{{\mathcal X}}
\newcommand{\yset}{{\mathcal Y}}
\newcommand{\zset}{{\mathcal Z}}
\newcommand{\cut}{\ominus}

\newcommand{\attractor}{A}

\newcommand{\statesx}{{\states^\px}}

\newcommand{\lossp}{\lambda}
\newcommand{\sglcstuple}{\tuple{\lcsstates,\lcsstatesz,\lcsstateso,\channels,\msgs,\lcstransitions,\lossp,\coloring}}
\newcommand{\sglcs}{{\mathcal L}}
\newcommand{\lcsstates}{{\tt S}}
\newcommand{\lcsstatesz}{\lcsstates^\pz}
\newcommand{\lcsstateso}{\lcsstates^\po}
\newcommand{\lcsstatesx}{\lcsstates^\px}

\newcommand{\lcsstate}{{\tt s}}
\newcommand{\msgs}{{\tt M}}
\newcommand{\msg}{{\tt m}}
\newcommand{\channels}{{\tt C}}
\newcommand{\channel}{{\tt c}}
\newcommand{\lcstransitions}{{\tt T}}
\newcommand{\lcstransition}{{\tt t}}
\newcommand{\op}{{\tt op}}
\newcommand{\nop}{{\tt nop}}

\newcommand{\chassignment}{{\tt x}}

\newcommand{\emptychannels}{\pmb{\emptyword}}

\newcommand{\transitionx}[1]{\overset{{#1}}{\transition}}

\newcommand{\ucof}[1]{{#1}\!\uparrow}
\newcommand{\ucset}{U}

\newcommand{\lcsstrat}{{\tt f}}
\newcommand{\induced}[1]{\underline{#1}}

\newcommand{\xlcsforcestrat}{{\tt force}^x}
\newcommand{\ylcsforcestrat}{{\tt force}^{1-x}}

\newcommand{\ylcsavoidstrat}{{\tt avoid}^{1-x}}
\newcommand{\xlcsstrat}{\lcsstrat^x}
\newcommand{\ylcsstrat}{\lcsstrat^{1-x}}
\newcommand{\xinducedlcsstrat}{\induced\lcsstrat^x}
\newcommand{\yinducedlcsstrat}{\induced\lcsstrat^{1-x}}

\begin{document}

\title[Stochastic Parity Games on LCS]{Stochastic Parity Games on
  Lossy Channel Systems\rsuper*}

\author[P.~A.~Abdulla]{Parosh Aziz Abdulla\rsuper a}
\address{{\lsuper{a,d}}Uppsala University, Department of Information
Technology, Box 337, SE-751 05 Uppsala, Sweden}
\urladdr{http://user.it.uu.se/\~parosh/}
\email{svens@it.uu.se}

\author[C.~Clemente]{Lorenzo Clemente\rsuper b}
\address{{\lsuper b}D{\'e}partement d'Informatique, Universit{\'e} Libre de Bruxelles (U.L.B.), Belgium}
\urladdr{https://sites.google.com/site/clementelorenzo}
%\email{clementelorenzo@gmail.com}

\author[R.~Mayr]{Richard Mayr\rsuper c}
\address{{\lsuper c}University of Edinburgh, School of Informatics,
10 Crichton Street, Edinburgh EH89AB, UK}
\urladdr{http://www.inf.ed.ac.uk/people/staff/Richard\_Mayr.html}
%\email{rmayr@inf.ed.ac.uk}

\author[S.~Sandberg]{Sven Sandberg\rsuper d}
%\address{}

%\thanks{Supported by} 

\titlecomment{{\lsuper*}Parts of this work have appeared in 
the proceedings of QEST 2013 \cite{SPGLCS:2013}.}

\begin{abstract}
We give an algorithm for solving
stochastic parity games with almost-sure winning conditions on
{\it lossy channel systems}, under the constraint that
both players are restricted to finite-memory strategies. 
First, we describe a general framework, where we
consider the class of 
$2\frac{1}{2}$-player games with
almost-sure parity winning conditions on possibly infinite game graphs, 
assuming that the game contains a {\it finite attractor}.
An attractor is a set of states (not necessarily absorbing)
that is almost surely re-visited regardless of the players' decisions.
We present a scheme that characterizes the set
of winning states for each player.
Then, we instantiate this scheme
to obtain an algorithm for
{\it stochastic game lossy channel systems}.
\end{abstract}

\maketitle

\section{Introduction}
\label{introduction:section}
\parg{Background.}
2-player games can be used to model the interaction of
a controller (player 0) who makes choices in a reactive 
system, and a malicious adversary (player 1) who represents
an attacker.
To model randomness in the system
(e.g., unreliability; randomized algorithms),
a third player `random' is defined who makes choices according
to a predefined probability distribution. The resulting 
stochastic game is called a $2\frac{1}{2}$-player game in the terminology of \cite{chatterjee03simple}.
The choices of the players induce a run of the system, and
the winning conditions of the game are expressed in terms of predicates 
on runs.

Most classic work on algorithms for stochastic games has focused
on finite-state systems (e.g.,
\cite{shapley-1953-stochastic,condon-1992-ic-complexity,AHK:FOCS98,chatterjee03simple}),
but more recently several classes of infinite-state systems have been
considered as well. 
Stochastic games on infinite-state probabilistic recursive systems (i.e.,
probabilistic pushdown automata with unbounded stacks) were studied in
\cite{Etessami:Yannakakis:ICALP05,EY:LMCS2008,EWY:ICALP08}.
A different (and incomparable) class of infinite-state systems are channel
systems, which use unbounded communication buffers instead of unbounded
recursion.

{\it Channel Systems} consist of nondeterministic
finite-state machines that communicate by asynchronous message passing
via unbounded FIFO communication channels. They are also known as
communicating finite-state machines (CFSM) \cite{Brand:CFSM}.
Channel Systems are a very expressive model that
can encode the behavior of Turing machines, by storing the content of an unbounded
tape in a channel \cite{Brand:CFSM}.
Therefore, all verification questions are undecidable on Channel Systems.

A {\it Lossy Channel System (LCS)} \cite{AbJo:lossy,Finkel:completely:specified} consists of
finite-state machines that communicate by asynchronous message passing
via unbounded \emph{unreliable} (i.e., lossy) FIFO communication channels,
i.e., messages can spontaneously disappear from channels.
The original motivation for LCS is to capture the behavior of communication protocols
which are designed to operate correctly even if the communication medium is unreliable
(i.e., if messages can be lost).
Additionally (and quite unexpectedly at the time),
the lossiness assumption makes safety/reachability and termination decidable \cite{AbJo:lossy,Finkel:completely:specified},
albeit of non-primitive recursive complexity \cite{schnoebelen-2002-ipl-verifying}.
However, other important verification problems are still undecidable for LCS,
e.g., recurrent reachability (i.e., B\"uchi properties), boundedness, and behavioural equivalences
\cite{AbJo:lossy:undecidable:journal,Schnoebelen:2001,Mayr:2003}.

A {\it Probabilistic Lossy Channel System (PLCS)}
\cite{Schnoeblen:plcs,Parosh:Alex:PLCS}
is a probabilistic variant 
of LCS where, in each computation step, each message
can be lost independently with a given probability.
This solves two limitations of LCS.
First, from a modelling viewpoint, probabilistic losses are more realistic
than the overly pessimistic setting of LCS where all messages can always be lost at any time.
Second, in PLCS almost-sure recurrent reachability properties become decidable (unlike for LCS) \cite{Schnoeblen:plcs,Parosh:Alex:PLCS}.
Several algorithms for symbolic model checking of PLCS have been presented 
\cite{Parosh:etal:attractor:IC,Rabinovich:plcs}.
The only reason why certain questions
are decidable for LCS/PLCS is that the message loss induces a quasi-order
on the configurations, which has the properties of a simulation.
Similarly to Turing machines and CFSM, 
one can encode many classes of infinite-state probabilistic transition
systems into a PLCS.
Some examples are:
\begin{itemize}
\item
Queuing systems where waiting customers in a queue drop out
with a certain probability in every time interval.
This is similar to the well-studied class of queuing systems
with impatient customers which practice {\em reneging}, i.e.,
drop out of a queue after a given maximal waiting time; see 
\cite{Wang-Li-Jiang:Review} section II.B.
Like in some works cited in \cite{Wang-Li-Jiang:Review}, the maximal waiting time
in our model is exponentially distributed.
In basic PLCS, unlike in \cite{Wang-Li-Jiang:Review}, this exponential distribution
does not depend on the current number of waiting customers.
However, an extension of PLCS with this feature would still 
be analyzable in our framework (except in the pathological case where
a high number of waiting customers increases the customers patience
exponentially, because such a system would not necessarily have a so-called \emph{finite attractor}; see below).
\item
Probabilistic resource trading games with
probabilistically fluctuating prices.
The given stores of resources are
encoded by counters (i.e., channels), which exhibit a probabilistic decline
(due to storage costs, decay, corrosion, obsolescence, etc).
\item
Systems modelling operation cost/reward, which is stored in counters/channels,
but probabilistically discounted/decaying over time.
\item
Systems which are periodically restarted (though not necessarily by a
deterministic schedule), due to, e.g., energy depletion or maintenance work. 
\end{itemize}
Due to this wide applicability of PLCS, we focus on this model in this paper.
However, our main results are formulated in more general terms referring
to infinite Markov chains with a finite attractor; see below.

\parg{Previous work.}

In \cite{BBS:ACM2007}, a non-deterministic extension of PLCS was introduced
where one player controls transitions in the control graph
and message losses are fully probabilistic.
This yields a Markov decision process (i.e., a $1\frac{1}{2}$-player game)
on the infinite graphs induced by PLCS.
It was shown in \cite{BBS:ACM2007}
that $1\frac{1}{2}$-player games with \emph{almost-sure}
repeated reachability (B\"uchi) objectives are decidable and pure memoryless determined.

In \cite{ABDMS:FOSSACS08}, $2\frac{1}{2}$-player games on PLCS are considered,
where the players control transitions in the control graph
and message losses are probabilistic.
Almost-sure B\"uchi objectives are decidable for this class,
and pure memoryless strategies suffice for \emph{both players} \cite{ABDMS:FOSSACS08}.
Generalized B\"uchi objectives are also decidable,
and finite-memory strategies suffice for the player,
while memoryless strategies suffice for the opponent \cite{BS-qapl2013}.

On the other hand, $1\frac{1}{2}$-player games on 
PLCS with \emph{positive probability}
B\"uchi objectives, i.e., almost-sure co-B\"uchi objectives from the
(here passive) opponent's point of view,
can require infinite memory to win
and are also undecidable \cite{BBS:ACM2007}. 
However, if the player is restricted to finite-memory strategies,
$1\frac{1}{2}$-player games with positive probability
\emph{parity objectives}
(even the more general \emph{Streett objectives})
become decidable and memoryless strategies suffice for the player \cite{BBS:ACM2007}.
Note that the finite-memory case and the infinite-memory one are a priori incomparable problems,
and neither subsumes the other.
Cf. Section~\ref{conclusions:section}.

Non-stochastic (2-player) parity games on infinite graphs were studied
in \cite{zielonka1998infinite}, where it is shown that such games are
determined, and that both players possess winning memoryless strategies in
their respective winning sets.  Furthermore,
a scheme for computing the winning
sets and winning strategies is given.
Stochastic games ($2\frac12$-player games) with parity conditions on
\emph{finite} graphs are known to be memoryless determined and
effectively solvable
\cite{alfaro-2000-lics-concurrent,chatterjee03simple,chatterjee-2006-qest-strategy}.

\parg{Our contribution.}
We give an algorithm to
decide almost-sure \emph{parity} games for probabilistic lossy channel systems
in the case where the players
are restricted to finite memory strategies.
We do that in two steps.
First, we give our result in general terms
(Section~\ref{parity:section}):
We
consider the class of 
$2\frac{1}{2}$-player games with
almost-sure parity wining conditions on possibly infinite game graphs, 
under the assumption
that the game contains a {\it finite attractor}.
An attractor is a set $\attractor$ of states 
such that, regardless of the strategies used by the players,
the probability measure of the runs which
visit $\attractor$ infinitely often is one.%
\footnote{
In the game community (e.g., \cite{zielonka1998infinite})
the word {\it attractor}
is used to
denote what we  call a {\it force set}
in Section~\ref{reachability:section}.
In the infinite-state systems community
(e.g., \cite{Parosh:etal:attractor:IC,Parosh:etal:MC:infinite:journal}), the word is used 
in the same way as we use it in this paper.}
Note that this means neither that $\attractor$ is absorbing, nor
that every run must visit $\attractor$.
We present a general scheme characterizing the set
of winning states for each player.
The scheme is a generalization of the well-known scheme for
non-stochastic games in \cite{zielonka1998infinite}.
In fact, the
constructions are equivalent in the case that no probabilistic states
are present.
We show correctness of the scheme for games where each player is
restricted to a finite-memory strategy.
The correctness proof here is more involved than in the
non-stochastic case of \cite{zielonka1998infinite};
we rely on the existence of a finite attractor and the restriction
of the players to use finite-memory strategies.
Furthermore, we show that if a player is winning against all
finite-memory strategies of the other player then he
can win using a \emph{memoryless} strategy.

In the second step (Section~\ref{sec:application:sglcs}),
we show that the scheme can be instantiated 
for lossy channel systems.
The above two steps yield an algorithm to
decide parity games in the case when the players
are restricted to finite memory strategies.
If the players are allowed infinite memory, then the problem
is undecidable already for $1\frac{1}{2}$-player games
with co-B\"uchi objectives (a special case of 2-color parity objectives)
\cite{BBS:ACM2007}.
Note that even if the players are restricted to finite memory strategies,
such a strategy (even a memoryless one) on an infinite game graph
is still an infinite object. Thus, unlike for finite game graphs,
one cannot solve a game by just guessing strategies
and then checking if they are winning.
Instead, we show how to
effectively compute a finite, symbolic representation of the
(possibly infinite) set of winning states for each player
as a regular language (Section~\ref{algorithm:section}),
and a finite description of winning strategies (Section~\ref{sec:strategy:construction}).

\section{Preliminaries}
\label{prels:section}

\parg{Notation.}
Let $\ord$ and $\nat$ denote the set of ordinal resp.\ natural numbers.
With $\alpha$, $\beta$, and $\gamma$ we denote arbitrary ordinals,
while with $\lambda$ we denote limit ordinals.
We use $f:X\to Y$ to denote that $f$ is a total function from $X$ to $Y$, and
use $f:X\partialto Y$ to denote that $f$ is a partial function from $X$ 
to $Y$.
We write $f(x)=\bot$ to denote that $f$ is undefined on $x$, and define
$\domof{f}:=\setcomp{x}{f(x)\neq\undef}$.
We say that $f$ is an {\it extension} of $g$ if $g(x)=f(x)$ whenever $g(x)\neq\undef$.
For $X'\subseteq X$, we use $f\restrict X'$ to denote the restriction of $f$ to $X'$.
We will sometimes need to pick an arbitrary element from a set.
To simplify the exposition, we let $\selectfrom{X}$ denote an arbitrary
but fixed element of the nonempty set $X$.

A \emph{probability distribution} on a countable set $X$ is a function
$f:X\to[0,1]$ such that $\sum_{x\in X}f(x)=1$.
For a set $X$, we use  $X^*$ and $X^\omega$ to denote the sets of finite
and infinite words over $X$, respectively.
The empty word is denoted by $\emptyword$.

\parg{Games.}
A \emph{game} (of \emph{rank $n$})
is a tuple $\game=\gametuple$ defined as follows.
$\states$ is a set of \emph{states}, partitioned into
the pairwise disjoint sets of \emph{random states} $\rstates$,
states $\zstates$ of Player$~\pz$, and states $\ostates$ of Player~$\po$.
$\transition\subseteq\states\times\states$ is the
  \emph{transition relation}.
  We write $\state\transition{}\state'$ to denote that
  $\tuple{\state,\state'}\in\transition{}$.
  We assume that for each $\state$ there is at least one and at most
  countably many $\state'$ with $\state\transition{}\state'$.
The \emph{probability function}
  $\probp:\rstates\times\states\to[0,1]$ satisfies both
  $\forall\state\in\rstates.\forall\state'\in\states .
  (\probp(\state,\state')>0\iff\state\transition\state')$ and
  $\forall\state\in\rstates .\sum_{\state'\in\states}
  \probp(\state,\state') = 1$.
(The sum is well-defined since we assumed that the number of successors of any state is at most countable.)
The \emph{coloring function} is defined as
  $\coloring:\states\to\{0,\dots,n\}$,
  where $\colorof\state$ is called the
  \emph{color} of state $s$.

Let 
  $\stateset\subseteq\states$ be a set of states.
  We use $\gcomplementof\game\stateset:=\states-\stateset$ to denote the 
  {\it complement} of $\stateset$.
  Define 
  $\zcut\stateset:=\stateset\cap\zstates$,
  $\ocut\stateset:=\stateset\cap\ostates$,
  $\zocut\stateset:=\zcut\stateset\cup\ocut\stateset$, and
  $\rcut\stateset:=\stateset\cap\rstates$.
  For  $n\in\nat$ and $\sim\;\in\{=,\leq\}$, let $\colorset
  \stateset\sim n:=\setcomp{\state\in \stateset}{\colorof\state\sim n}$
  denote the sets of states in $\stateset$ with color $\sim n$.

    A \emph{run} $\run$ in $\game$ is an infinite sequence
$\state_0\state_1\cdots$ of states s.t.
$\state_i\transition{}\state_{i+1}$ for all $i\geq 0$;
$\run(i)$ denotes $\state_i$.
A \emph{path} $\pth$ is a finite sequence $\state_0\cdots\state_n$ of
states s.t. $\state_i\transition{}\state_{i+1}$ for all $i:0\leq
i<n$.
We say that $\run$ (or $\pth$) \emph{visits} $\state$ if
$\state=\state_i$ for some $i$.
For any $\stateset\subseteq\states$, we use $\pthset{\stateset}$ to
denote the set of paths that end in some state in $\stateset$.
Intuitively, the choices of the players and the resolution of randomness induce a run
$\state_0\state_1\cdots$, 
starting in some initial state $\state_0\in\states$;
state $\state_{i+1}$ is chosen as a successor of $\state_i$,
and this choice is made by Player~0 if $\state_i\in\zstates$,
by Player~1 if $\state_i\in\ostates$,
and it is chosen randomly according to the probability distribution
$\probp(\state_i,\cdot)$ if $\state_i\in\rstates$. 

\parg{Strategies.}
For $\px\in\set{0,1}$,
a strategy for Player~$\px$ prescribes the next move,
given the current prefix of the run.
Formally, a \emph{strategy} of Player~$\px$ is a partial
function $\xstrat:\pthset{\xstates}\partialto\states$ s.t.
$\state_n\transition\xstrat(\state_0\cdots\state_{n})$ if 
$\xstrat(\state_0\cdots\state_{n})$ is
defined.
The strategy $\xstrat$ prescribes for Player~$\px$ the next move,
given the current prefix of the run.
A run $\run=\state_0\state_1\cdots$ is said to be {\it consistent}
with a strategy $\xstrat$ of Player~$\px$ if 
$\state_{i+1}=\xstrat(\state_0\state_1\cdots\state_i)$
whenever $\xstrat(\state_0\state_1\cdots\state_i)\neq\bot$.
We say that $\run$ is {\it induced} by $\tuple{\state,\xstrat,\ystrat}$
if $\state_0=\state$ and $\run$ is consistent with both $\xstrat$ and $\ystrat$.
We use $\runsof{\game,\state,\xstrat,\ystrat}$ to denote the set of runs
in $\game$ induced by 
$\tuple{\state,\xstrat,\ystrat}$.
We say that $\xstrat$ is total if it is defined for every
$\pth\in\pthset{\xstates}$.

A strategy $\xstrat$ of Player~$\px$ is \emph{memoryless}
  if the next state only depends on the current state and not on the
  previous history of the run, i.e., for any path
  $\state_0\cdots\state_n\in\pthset{\xstates}$, we have
  $\xstrat(\state_0\cdots\state_n)=\xstrat(\state_n)$.

  A \emph{finite-memory strategy} updates a finite memory
  each time a transition is taken, and the next state depends only on
  the current state and memory.
  Formally, we define a \emph{memory structure} for Player~$\px$ as a quadruple
  $\memstratn=\memstrattuple$ satisfying the following properties.
  The nonempty set $\memory$ is called the \emph{memory} and
  $\memconf_0\in\memory$ is the \emph{initial memory
    configuration}.
  For a current memory configuration $m$ and a current state $s$, the next
  state is given by 
  $\memtrans : \xstates\times\memory \to\states$, where 
  $\state\transition\memtrans(\state,\memconf)$.
  The next memory configuration is given by
  $\memmem:\states\times\memory\to\memory$.
  We extend $\memmem$ to paths by
  $\memmem(\emptyword,\memconf)=\memconf$ and
  $\memmem(\state_0\cdots\state_n,\memconf) =
  \memmem(\state_n,\memmem(\state_0\cdots\state_{n-1},\memconf))$.
  The total strategy $\memstratstratn:\pthset{\xstates}\to\states$
  induced by $\memstratn$ is given by $
  \memstratstratn(\state_0\cdots\state_{n})
  :=\memtrans(\state_{n},\memmem(\state_0\cdots\state_{n-1},\initmem))
  $.
  A total strategy $\xstrat$ is said to have \emph{finite memory} if
  there is a memory structure $\memstratn=\memstrattuple$ where
  $\memory$ is finite and $\xstrat=\memstratstratn$.
Consider a run $\run=\state_0\state_1\cdots\in\runsof{\game,\state,\xstrat,\ystrat}$ where $\ystrat$ is induced by $\memstratn$.
We say that $\run$ \emph{visits} the configuration
$\tuple{\state,\memconf}$ if there is an
$i$ such that $\state_i=\state$ and
$\memmem(\state_0\state_1\cdots\state_{i-1},\memconf_0)=\memconf$.

We use $\xallstrats(\game)$, $\xfinitestrats(\game)$, and $\xnomemstrats(\game)$ 
to denote the set of {\it all}, {\it finite-memory}, and
{\it memoryless} strategies respectively of Player~$\px$ in $\game$.
Note that memoryless strategies and strategies in general can be
partial, whereas for simplicity we only define total finite-memory
strategies.
 
\parg{Probability Measures.}
We use the standard definition of probability measures for a set of runs
\cite{billingsley-1986-probability}.
First, we define the measure for total strategies, and then we extend 
it to general (partial) strategies.
Consider a game $\game=\gametuple$, an initial state $\state$, and total
strategies $\xstrat$ and $\ystrat$ of Players~$\px$ and~$\py$.
Let
$\Om^{\state}=\state\states^{\om}$ denote the set of all infinite
sequences of states starting from $\state$.
For a measurable set 
${\runset}\subseteq\Om^\state$, 
we define $\probm_{\game,\state,\xstrat,\ystrat}({\runset})$ to be
the probability measure of $\runset$ under the strategies $\xstrat,\ystrat$.
This measure is well-defined
\cite{billingsley-1986-probability}.
%
%When the state $\state$ is known from context, we drop the superscript
%and write $\ProbmStateStrat{\xstrat,\ystrat}(\runset)$.
%
For (partial) strategies 
$\xstrat$ and $\ystrat$ of Players~$\px$ and~$\py$,
$\sim\;\in\{<,\leq,=,\geq,>\}$,
a real number $c\in[0,1]$,
and any measurable set $\runset\subseteq\Om^\state$, 
we define $\probm_{\game,\state,\xstrat,\ystrat}({\runset}) \sim c$
iff $\probm_{\game,\state,g^\px,g^{\py}}(\runset)\sim c$ for all total
strategies $g^\px$ and $g^{\py}$ that are extensions of
$\xstrat$ resp.\ $\ystrat$.

\parg{Winning Conditions.}

The winner of the game is determined by a predicate on
infinite runs.
We assume familiarity with the syntax and semantics of the temporal
logic ${\mathit CTL}^*$ (see, e.g., \cite{CGP:book}).
Formulas are interpreted on the structure $(\states,\transition)$.
We use 
$\denotationof{\formula}{\state}$ to denote the set of runs starting from
$\state$ that satisfy the ${\mathit CTL}^*$ path-formula $\formula$.
This set is measurable \cite{Vardi:probabilistic},
and we just write
$\probm_{\game,\state,\xstrat,\ystrat}(\formula) \sim c$
instead of 
$\probm_{\game,\state,\xstrat,\ystrat}(\denotationof\formula\state) \sim c$.

We will consider games with \emph{parity}
winning conditions,
whereby Player~1 wins if the largest color that occurs infinitely often in the
infinite run is odd, and Player~0 wins if it is even.
Thus, the winning condition for Player~$\px$ can be expressed in  
${\mathit CTL}^*$ 
as 
\[
 \xparity :=
 \bigvee_{i\in \{0,\dots,n\}\wedge (i\bmod 2)=x}(
 \always\eventually \colorset \states=i \wedge
 \eventually\always \colorset \states\leq i) \ .
\]

\parg{Winning Sets.}

For a strategy $\xstrat$ of Player~$\px$,
and a set $\ystratset$ of strategies of Player~$\py$,
we define
\begin{align*}
\xwinset(\xstrat,\ystratset)(\game,\formula^{\sim c}):=
\setcomp{\state}{\forall \ystrat\in\ystratset .
\ystrat\text{ is total} \implies
\probm_{\game,\state,\xstrat,\ystrat}(\formula)\sim c}
\end{align*}
If there is a strategy $\xstrat$ such that $\state\in\xwinset(\xstrat,\ystratset)(\game,\formula^{\sim c})$,
then we say that $\state$ is a {\it winning state}
for Player~$\px$ in $\game$
wrt.\ $\formula^{\sim c}$ (and $\xstrat$ is \emph{winning at $\state$}),
provided that Player~$\py$
is restricted to strategies in
$\ystratset$.
Sometimes, when the parameters $\game$, $\state$, 
$\ystratset$,  $\formula$, and $\sim c$ are known, we will not mention them and
may simply say that ``$\state$ is a winning state'' or that
``$\xstrat$ is a winning strategy'', etc.
If 
$\state\in\xwinset(\xstrat,\ystratset)(\game,\formula^{=1})$,
then
we say that Player~$\px$
wins from $\state$ \emph{almost surely (\as)}.
If
$\state\in\xwinset(\xstrat,\ystratset)(\game,\formula^{>0})$,
then
we say that Player~$\px$ wins from $\state$ \emph{with positive probability (\wpp)}.

We also define 
$\xvinset(\xstrat,\ystratset)(\game,\formula):=
\setcomp{\state}{\forall \ystrat\in\ystratset.\;\runsof{\game,\state,\xstrat,\ystrat}\subseteq\denotationof\formula\state}$.
If $\state\in\xvinset(\xstrat,\ystratset)(\game,\formula)$, then we say that
Player~$\px$ {\it surely} wins from $\state$.
Notice that any strategy that is surely winning from a state $\state$
is also winning from $\state$ \as, and any strategy that is winning \as\  is also winning \wpp, i.e.,
$\xvinset(\xstrat,\ystratset)(\game,\formula)\subseteq
\xwinset(\xstrat,\ystratset)(\game,\formula^{=1})\subseteq
\xwinset(\xstrat,\ystratset)(\game,\formula^{>0})$.

\parg{Determinacy and Solvability.}
A game is called \emph{determined} wrt. an objective $\formula^{\sim c}$ and
two sets $\zstratset,\ostratset$ of strategies of Player~$\pz$, resp.\ Player~$\po$,
if, for every state $s$,
Player~$\px$ has a strategy $\xstrat\in\xstratset$ that is winning against all strategies $g\in\ystratset$ of the opponent,
i.e., $s \in \xwinset(\xstrat,\ystratset)(\game,\textrm{cond}_x)$,
where $\textrm{cond}_0 = \formula^{\sim c}$ and $\textrm{cond}_1 = \formula^{\not\sim c}$.
By \emph{solving} a determined game, we mean giving an algorithm to
compute symbolic representations of the sets of states which are 
winning for either player and a symbolic representation of the corresponding winning strategies.

\parg{Attractors.}
A set $\attractor\subseteq\states$ is said
to be an {\it attractor} if,
for each state $\state\in\states$ and strategies
$\zstrat,\ostrat$ of Player~$\pz$ resp.\ Player~$\po$,
 it is the case that
$\probm_{\game,\state,\zstrat,\ostrat}(\eventually\attractor)=1$.
In other words, regardless of where
we start a run and regardless of the strategies
used by the players, we will
reach a state inside the attractor \as.
It is straightforward to see that this also implies that
$\probm_{\game,\state,\zstrat,\ostrat}(\always\eventually\attractor)=1$,
i.e., the attractor will be visited infinitely often \as

\parg{Transition Systems.}
Consider strategies $\xstrat\in\xnomemstrats$ and 
$\ystrat\in\yfinitestrats$ of
Player~$\px$ resp.\ Player~$\py$,
where $\xstrat$ is memoryless and
$\ystrat$ is finite-memory.
Suppose that $\ystrat$ is 
induced by memory structure   $\memstratn=\memstrattuple$.
We define the {\it transition system}
$\tsys$
induced
by $\game,\ystrat,\xstrat$ to be the pair
$\tuple{\memstates,\tmovesto}$
where
$\memstates=\states\times\memory$, and
$\tmovesto\subseteq\memstates\times\memstates$ such that
$\tuple{\state_1,\memconf_1}\tmovesto\tuple{\state_2,\memconf_2}$ if
$\memconf_2=\memmem(\state_1,\memconf_1)$, and
one of the following three conditions is satisfied:
(i) 
$\state_1\in\xstates$ and
either $\state_2=\xstrat(\state_1)$ or $\xstrat(\state_1)=\bot$,
(ii)  
$\state_1\in\ystates$ and $\state_2=\memtrans(\state_1,\memconf_1)$, or
(iii) 
$\state_1\in\rstates$ and $\probp(\state_1,\state_2)>0$.

Consider the directed acyclic graph (DAG) of maximal strongly
connected components (SCCs) of the transition system 
$\tsys$.
An SCC is called a {\it bottom SCC (BSCC)} if no other SCC is 
reachable from
it.  
Observe that the existence of BSCCs is not guaranteed in an
infinite transition system.
However, if $\game$ contains a finite attractor $\attractor$ 
and $\memory$ is finite
then $\tsys$ contains at least one BSCC, and  in fact
each BSCC contains at least one element
$\tuple{\state_\attractor,\memconf}$
with $\state_\attractor\in\attractor$.
In particular, for any state $\state\in\states$,
any run $\run\in\runsof{\game,\state,\xstrat,\ystrat}$ will visit
a configuration $\tuple{\state_\attractor,\memconf}$
infinitely often \as\ 
where $\state_\attractor\in\attractor$
and $\tuple{\state_\attractor,\memconf}\in\bscc$ 
for
some BSCC $\bscc$.

\section{Reachability}
\label{reachability:section}

In this section we present some concepts related
to checking reachability objectives in games.
First, we define basic notions.
Then we recall a standard scheme
(described e.g. in \cite{zielonka1998infinite})
for
checking reachability winning conditions, and
state some of its properties that we use
in the later sections.
In this section, we do not use the finite attractor property,
nor do we restrict the class of strategies in any way.
Below, fix a  game $\game=\gametuple$.

\parg{Reachability Properties.}
  Fix a state $\state\in\states$ and sets of states 
  $\stateset,\stateset'\subseteq\states$.
Let $\postof\game{\state}:=\{\state':\state\transition\state'\}$ denote the
  set of \emph{successors} of $\state$. Extend it to sets of states  by
  $\postof\game{\stateset}:=\bigcup_{\state\in\stateset}\postof\game{\state}$.
  Note that for any given state $\state\in\rstates$,
  $\probp(\state,\cdot)$ is a probability distribution over
  $\postof\game{\state}$.
  Let $\preof\game{\state}:=\{\state':\state'\transition\state\}$ denote the
  set of \emph{predecessors} of $\state$, and extend it to sets of
  states as above.
  We define $\dualpreof\game{\stateset}:=\gcomplementof\game{\preof\game{\gcomplementof\game\stateset}}$, 
  i.e., it denotes the set of  states whose successors 
  {\it all} belong to $\stateset$.
  We say that $\stateset$ is \emph{sink-free} if
  $\postof\game{\state}\cap\stateset\neq\emptyset$ for all $\state\in\stateset$,
  and \emph{closable} if it is sink-free and 
  $\postof\game\state\subseteq\stateset$ for all $\state\in\rcut\stateset$.
  If $\stateset$ is closable then each state in $\zocut\stateset$ 
  has at least   one successor in $\stateset$, and all 
  the successors of states in $\rcut\stateset$ are  in $\stateset$.

  For $\px\in\set{0,1}$, we say that $\stateset$ is an \emph{$\px$-trap} if it is closable and $\postof\game{\state}\subseteq\stateset$
  for all $\state\in\xcut\stateset$.
  Notice that $\states$ is both a $\pz$-trap and a $\po$-trap, and 
  in particular it is both sink-free and closable.
The following lemma states that, starting from a state inside a set of states 
$\stateset$  that is a trap for one player, 
the other player can surely keep the run inside $\stateset$.
\begin{lem}
\label{trap:certainly:lemma}
If $\stateset$ is a $(\py)$-trap, then there exists
a memoryless strategy $\xstrat\in\xnomemstrats(\game)$ for Player~$\px$ 
such that
$\stateset\subseteq\xvinset(\xstrat,\yallstrats(\game))(\game,\always\stateset)$.
\end{lem}
\begin{proof}

We define a memoryless strategy $\xstrat$ of Player~$\px$ that is 
surely winning
from any state $\state\in\stateset$, i.e.,
$\stateset\subseteq\xvinset(\xstrat,\yallstrats(\game))(\game,\always\stateset)$.
For a state $\state\in\xcut\stateset$, we define
$\xstrat(\state)=\selectfrom{\postof\game{\state}\cap\stateset}$.
This is well-defined since $\stateset$
is a $(\py)$-trap.
We can now show that any run that starts
from a state $\state\in\stateset$ and
that is consistent with 
$\xstrat$ will surely remain inside
$\stateset$.
Let $\ystrat$ be any strategy of Player~$\py$, and
let $\state_0\state_1\ldots\in\runsof{\game,\state,\xstrat,\ystrat}$.
We show, by induction on $i$,
that $\state_i\in\stateset$ for all $i\geq 0$.
The base case is clear since $\state_0=\state\in\stateset$.
For $i>1$, we consider three cases depending on $\state_i$:
\begin{itemize}
\item
$\state_i\in\xcut\states$.
By the induction hypothesis we know that 
$\state_i\in\stateset$,
and hence by definition of $\xstrat$ we know that 
$\state_{i+1}=\xstrat(\state_i)\in\stateset$.
\item
$\state_i\in\ycut\states$.
By the induction hypothesis we know that 
$\state_i\in\stateset$,
and hence $\state_{i+1}\in\stateset$
since $\stateset$ is a $(\py)$-trap.
\item
$\state_i\in\rcut\states$.
By the induction hypothesis we know that 
$\state_i\in\stateset$,
and hence $\state_{i+1}\in\stateset$
since $\stateset$ is closable. \qedhere
\end{itemize}
\end{proof}

\parg{Scheme.}
Given a set $\targetset\subseteq\states$,  we give a scheme for 
 computing a partitioning of $\states$ into two sets
$\xforceset(\game,\targetset)$ and $\yavoidset(\game,\targetset)$
s.t. 1) Player~$\px$ has a memoryless strategy on $\xforceset(\game,\targetset)$ to force the game to $\targetset$ \wpp,
and 2) Player~$\py$ has a memoryless strategy on $\yavoidset(\game,\targetset)$ to surely avoid $\targetset$.
The scheme and its correctness is adapted from \cite{zielonka1998infinite} to the stochastic setting.

First, we characterize the states that are winning for Player~$\px$,
by defining an increasing set of 
states each of which consists of winning states
for Player~$\px$, as follows:%
\begin{align*}
\reachset_0&:=\targetset\\
\reachset_{\alpha+1}&:=\reachset_\alpha\cup
\rcut{\preof\game{\reachset_\alpha}}\cup
\xcut{\preof\game{\reachset_\alpha}}\cup
\ycut{\dualpreof\game{\reachset_\alpha}} \\
\reachset_\lambda&:=\bigcup_{\alpha<\lambda}\reachset_\alpha
 \qquad \textrm { (for $\lambda$ a limit ordinal) }
\end{align*}
Clearly, the sequence is non-decreasing, i.e., $\reachset_\alpha\subseteq\reachset_\beta$ when $\alpha \leq \beta$,
and since the sequence is bounded by $S$, it converges at some (possibly infinite) ordinal.
We state this as a lemma:
\begin{lem}
\label{reachability:tp:lemma}
There is a $\tp\in\ord$ such that
$\reachset_\tp=\bigcup_{\alpha\in\ord}\reachset_\alpha$.
\end{lem}
\noindent
Let $\tp$ be the smallest ordinal s.t. $\reachset_\tp = \reachset_{\tp+1}$ (it exists by the lemma above).
We define 
\begin{align*}
	\xforceset(\game,\targetset)&:=\reachset_\tp\\
	\yavoidset(\game,\targetset)&:=\;\gcomplementof\game{\reachset_\tp}
\end{align*}

\begin{lem}
\label{not:reach:trap:lemma}
$\yavoidset(\game,\targetset)$ is an $\px$-trap.
\end{lem}
\begin{proof}
	Recall that $\yavoidset(\game,\targetset)=\gcomplementof\game{\reachset_\tp}$
	and $\reachset_{\tp+1}\subseteq\reachset_\tp$.
First, we prove that $\gcomplementof\game\reachset_\tp$ is sink-free.
There are two cases to consider:
\begin{itemize}
\item
$\state\in\xcut{\gcomplementof\game\reachset_\tp} \cup \rcut{\gcomplementof\game\reachset_\tp}$.
First, $\postof\game\state\subseteq\;\gcomplementof\game\reachset_\tp$.
Indeed, if not, we would have $\postof\game\state\cap\reachset_\tp\neq\emptyset$,
and thus $\state\in\reachset_{\tp+1}\subseteq\reachset_\tp$,
which is a contradiction.
Second, since $\states$ is sink-free, we have
$\postof\game\state\neq\emptyset$,
and thus $\postof\game\state\cap\gcomplementof\game\reachset_\tp\neq\emptyset$.
\item
$\state\in\ycut{\gcomplementof\game\reachset_\tp}$.
We clearly have $\postof\game\state\cap\gcomplementof\game\reachset_\tp\neq\emptyset$,
otherwise $\postof\game\state\subseteq\reachset_\tp$,
and thus $\state\in\reachset_{\tp+1}\subseteq\reachset_\tp$,
which is a contradiction.
\end{itemize}
Second, when proving sink-freeness above, we showed that
$\postof\game\state\subseteq\;\gcomplementof\game\reachset_\tp$
for any $\state\in\rcut{\gcomplementof\game\reachset_\tp}$ 
which means that $\gcomplementof\game\reachset_\tp$  
is closable.
Finally, we also showed
that $\postof\game\state\subseteq\;\gcomplementof\game\reachset_\tp$
for any $\state\in\xcut{\gcomplementof\game\reachset_\tp}$,
which means that $\gcomplementof\game\reachset_\tp$ is an $\px$-trap,
thus concluding the proof.
\end{proof}

The following lemma shows correctness of the construction.
In fact, it shows that a winning player also has a
memoryless strategy which is winning against an arbitrary opponent.
\begin{lem}
\label{reachability:correct:lemma}
There are memoryless strategies $\xforcestrat(\game,\targetset)\in\xnomemstrats(\game)$ for Player~$\px$
and \\$\yavoidstrat(\game,\targetset)\in\ynomemstrats(\game)$ for Player~$\py$ s.t.
\begin{align*}
 \xforceset(\game,\targetset)&\subseteq
	\xwinset(\xforcestrat(\game,\targetset),\yallstrats(\game))(\game,\eventually\targetset^{>0}) \\
 \yavoidset(\game,\targetset)&\subseteq
	\yvinset(\yavoidstrat(\game,\targetset),\xallstrats(\game))(\game,\always(\gcomplementof\game\targetset))
\end{align*}
\end{lem}
\begin{proof}
Let $\reachset=\xforceset(\game,\targetset)$.
To prove the first claim, we define a memoryless strategy $\xstrat$ of Player~$\px$ that is winning
from $\reachset$.
For any $\state\in\xcut{\reachset}$,
let $\alpha$ be the unique ordinal s.t. $\state\in\xcut{\reachset_{\alpha+1}\setminus\reachset_\alpha}$.
Then, we define 
$\xstrat(\state):=\selectfrom{\postof\game{\state}\cap\reachset_\alpha}$.
We show that $\xstrat$ forces the run to the target set $\targetset$ w.p.p. against an arbitrary opponent.
Fix a strategy $\ystrat$ for Player~$\py$.
We show that 
$\probm_{\game,\state,\xstrat,\ystrat}(\eventually\targetset)>0$ by transfinite induction.
If $\state\in\reachset_0$, then the claim follows trivially.
If $\state\in\reachset_{\alpha+1}$, then either
$\state\in\reachset_{\alpha}$
in which case the claim holds by the induction hypothesis,
or $\state\in\reachset_{\alpha+1}\setminus\reachset_{\alpha}$.
In the latter case, there are three sub-cases:
\begin{itemize}
\item
  $\state\in\xcut{\reachset_{\alpha+1}\setminus\reachset_\alpha}$.
  By definition of $\xstrat$, 
  we know that $\xstrat(\state)=\state'$ for some $\state'\in\reachset_{\alpha}$.
  By the induction hypothesis,
  $\probm_{\game,\state',\zstrat,\ostrat}(\eventually\targetset)>0$, and hence
  $\probm_{\game,\state,\zstrat,\ostrat}(\eventually\targetset)>0$.
\item
  $\state\in\ycut{\reachset_{\alpha+1}\setminus\reachset_{\alpha}}$.
  Let $s'$ be the successor of $s$ chosen by $\ystrat$.
  By definition of $\reachset_{\alpha+1}$,
  we know that $\state'\in\reachset_{\alpha}$.
  Then, the proof follows as in the previous case.
\item
  $\state\in\rcut{\reachset_{\alpha+1}\setminus\reachset_{\alpha}}$.
  By definition of $\reachset_{\alpha+1}$,
  there is a $\state'\in\reachset_{\alpha}$
  such that $\probp(\state,\state')>0$.
  By the induction hypothesis,
  $\probm_{\game,\state,\zstrat,\ostrat}(\eventually\targetset)
  \geq 
  \probm_{\game,\state',\zstrat,\ostrat}(\eventually\targetset)\cdot\probp(\state,\state')>0$.
\end{itemize}
Finally, if $\state\in\reachset_{\lambda}$
for a limit ordinal $\lambda$, then
$\state\in\reachset_{\alpha}$ 
for some $\alpha<\lambda$,
and the claim follows by the induction hypothesis.

From Lemma~\ref{not:reach:trap:lemma} and
Lemma~\ref{trap:certainly:lemma} it follows that
there is a strategy $\ystrat$ for Player~$\py$ such that
$\yavoidset(\game,\targetset)\subseteq\yvinset(\ystrat,\xallstrats)(\game,\always(\yavoidset(\game,\targetset)
))$.
The second claim follows then from the fact that 
$\targetset\cap \yavoidset(\game,\targetset)=\emptyset$.
\end{proof}

\section{Parity Conditions}
\label{parity:section}
We describe a scheme for solving stochastic parity games with
almost-sure winning conditions on infinite graphs,
under the conditions that the game has a finite attractor (as
defined in Section~\ref{prels:section}),
and that the players are restricted to finite-memory strategies.

We define a sequence of functions $\cset_0,\cset_1,\ldots$
Each $\cset_n$ takes a single argument, a game of rank at most $n$,
and it returns the set of states where Player~$\px$ wins \as, with
$\px=n\bmod 2$.
In other words, the player that has the same parity as color $n$ wins \as\  in
$\cset_n(\game)$.
We provide a memoryless strategy that is winning \as~for Player~$\px$
in $\cset_n(\game)$ against any finite-memory strategy of Player~$\py$,
and a memoryless strategy that is winning \wpp~for Player~$\py$ in
$\gcomplementof \game \cset_n(\game)$ against any finite-memory strategy of
Player~$\px$.

The scheme is by induction on $n$ and is related to \cite{zielonka1998infinite}.
In the rest of the section, we make use of the following notion of sub-game.
For a closable $\gcomplementof\game\stateset$, we define the \emph{sub-game}
$\game\cut\stateset:=\tuple{\stateset',\zcut{\stateset'},\ocut{\stateset'},
\rcut{\stateset'},\transition',\probp',\coloring'}$, where
$\stateset':=\gcomplementof\game\stateset$ is the new set of states,
$\transition':=\transition\cap({\stateset'}\times{\stateset'})$, 
$\probp':=\probp\restrict(\rcut{\stateset'}\times{\stateset'})$, 
$\coloring':=\coloring\restrict{\stateset'}$.
Notice that $\probp'(\state)$ is a probability 
distribution for any $\state\in\rcut{\stateset'}$ since $\stateset'$ is closable.
We use $\game\cut\stateset_1\cut\stateset_2$ to denote 
$(\game\cut\stateset_1)\cut\stateset_2$.

For the base case, let $\cset_0(\game):=\states$
for any game $\game$ of rank 0.
Indeed, from any configuration Player 0 trivially wins \as\ (even surely) because there is only color 0.

\input{fig_parity_schema}

For $n\geq 1$, let $\game$ be a game of rank $n$.
In the following, let \[ \px=n\bmod 2 .\]
$\cset_n(\game)$ is defined with the help of two auxiliary transfinite sequences
of sets of states $\set{\xset_\alpha}_{\alpha\in\ord}$ and $\set{\yset_\alpha}_{\alpha\in\ord}$.
The construction ensures that
$\xset_0\subseteq\yset_0\subseteq\xset_1\subseteq\yset_1\subseteq\cdots$, 
and that the states of $\xset_\alpha,\yset_\alpha$ are winning \wpp\ 
for Player~$\py$.
We use strong induction, i.e., to construct $\xset_\alpha$ we assume that
$\xset_\beta$ has been constructed for all $\beta<\alpha$, and it suffices to state
one unified inductive step rather than distinguishing between base
case, successor ordinals and non-zero limit ordinals.
In the (unified) inductive step, we have already constructed
$\xset_\beta$ and $\yset_\beta$ for all $\beta<\alpha$.
Our construction of $\xset_\alpha$ and $\yset_\alpha$ is in three steps (cf. Figure~\ref{fig:parity:schema}):
\begin{enumerate}
\item $\xset_\alpha$ is the set of states
  where Player~$\py$ can force the run to visit
  $\bigcup_{\beta<\alpha}\yset_\beta$
  \wpp
\item Find a set of states where
  Player~$\py$ wins \wpp\  in
  the sub-game $\game\cut\xset_\alpha$.
\item Take $\yset_\alpha$ to be the union of $\xset_\alpha$ and the set constructed in step 2.
\end{enumerate}
We next show how to find the winning states in 
the sub-game $\game\cut\xset_\alpha$ in step 2.
We first compute the set of states where Player~$\px$ can force the play
in $\game\cut\xset_\alpha$ to reach a state with color $n$ 
\wpp
We call this set $\zset_\alpha$.
The sub-game $\game\cut\xset_\alpha\cut\zset_\alpha$ does not contain any states of color $n$.
Therefore, this game can be completely solved, using the already constructed
function $\cset_{n-1}(\game\cut\xset_\alpha\cut\zset_\alpha)$.
The resulting winning set is winning \as\  in
$\game\cut\xset_\alpha\cut\zset_\alpha$, hence it is winning \wpp
We will prove that the states where Player~$\py$ wins \wpp\  in 
$\game\cut\xset_\alpha\cut\zset_\alpha$ are winning
\wpp\  also in $\game$.
We thus take $\yset_\alpha$ as the union of $\xset_\alpha$ and 
$\cset_{n-1}(\game\cut\xset_\alpha\cut\zset_\alpha)$.

We define the sequences formally:
%%% DISPLAYED MATH
\begin{align*}
	\xset_\alpha &:= \yforceset(\game,{\textstyle\bigcup_{\beta<\alpha}\yset_\beta}) \\
    \zset_\alpha &:= \xforceset(\game\cut\xset_\alpha,\colorset{\gcomplementof\game\xset_\alpha}=n) \\
    \yset_\alpha &:= \xset_\alpha\cup\cset_{n-1}(\game\cut\xset_\alpha\cut\zset_\alpha)
\end{align*}
%%% END MATH
%
Notice that the sub-games $\game\cut\xset_\alpha$ and
$\game\cut\xset_\alpha\cut\zset_\alpha$ are well-defined, since
$\gcomplementof\game\xset_\alpha$ is closable in $\game$ (by Lemma~\ref{not:reach:trap:lemma}), and 
$\gcomplementof{\game\cut\xset_\alpha}\zset_\alpha$ is closable in $\game\cut\xset_\alpha$.

By the definition, for $\alpha \leq \beta$ we get
$\yset_\alpha \subseteq \xset_\beta \subseteq \yset_\beta$.
As in Lemma~\ref{reachability:tp:lemma},
we can prove that this sequence converges:
\begin{lem}
\label{tp:lemma}
There exists a $\tp\in\ord$ such that
$\xset_\tp = \yset_\tp = \bigcup_{\alpha\in\ord}\yset_\alpha$.
% and (ii) $\cset_n(\game)=\gcomplementof\game\xset_\tp$.
\end{lem}
\noindent
Let $\tp$ be the least ordinal s.t. $X_{\tp + 1} = X_{\tp}$ (which exists by the lemma above).
We define
\begin{align}
	\label{eq:cset:def}
	\cset_n(\game):=\gcomplementof\game{\xset_\gamma}
\end{align}
The following lemma shows the correctness of
the construction.
Recall that we assume that $\game$ is of rank $n$ and that it contains
a finite attractor.
\begin{lem}
\label{cn:infinite:termination:lemma}
There are memoryless strategies
$\xstrat_c\in\xnomemstrats(\game)$ for Player~$\px$ and
$\ystrat_c\in\ynomemstrats(\game)$ for Player~$\py$
such that the following two properties hold:
\begin{align}
	\cset_n(\game)\;&\subseteq\;\xwinset(\xstrat_c,\yfinitestrats(\game))(\game,{\xparity}^{=1} ) \\
	\gcomplementof\game{\cset_n(\game)}\;&\subseteq\;\ywinset(\ystrat_c,\xfinitestrats(\game))(\game,{\yparity}^{>0} )
\end{align}
\end{lem}
\begin{proof}
Using induction on $n$,
we define the strategies $\xstrat_c,\ystrat_c$, 
and prove that the strategies are indeed winning.

%%%%%%%%
\parg{Construction of $\xstrat_c$.}
For $n\geq 1$, recall that $\tp$ is the least ordinal s.t. $\xset_{\tp+1} = \xset_\tp$ (as defined above),
and define $\complementof{\xset_\tp}:=\gcomplementof\game{\xset_\tp}$ and 
$\complementof{\zset_\tp}:=\gcomplementof\game{\zset_\tp}$.
By definition, $\cset_n(\game)=\complementof{\xset_\tp}$.
For a state $\state\in\complementof{\xset_\tp}$, we define $\xstrat_c(\state)$ depending
on the membership of $\state$ in one of the following three partitions
of $\complementof{\xset_\tp}$:
\begin{enumerate}
\item
$\state\in \complementof{\xset_\tp}\cap\complementof{\zset_\tp}$.
Define $\game':=\game\cut\xset_\tp\cut\zset_\tp$.
By the definition of $\tp$, we have that
$\xset_{\tp+1}\setminus\xset_\tp=\emptyset$.
By the construction of $\yset_\alpha$ we have, for an arbitrary $\alpha$,
that $\cset_{n-1}(\game\cut\xset_\alpha\cut\zset_\alpha)=\yset_\alpha\setminus\xset_\alpha$,
and by the construction of $\xset_{\alpha+1}$, we have that
$\yset_\alpha\setminus\xset_\alpha\subseteq\xset_{\alpha+1}\setminus\xset_\alpha$.
By combining these facts, we obtain
$\cset_{n-1}(\game')\subseteq\xset_{\tp+1}\setminus\xset_\tp=\emptyset$.
Since $\game\cut\xset_\tp\cut\zset_\tp$ does not contain any states of
color $n$
(or higher), it follows by the induction hypothesis that there is a
memoryless strategy $\strat_1\in\xnomemstrats(\game')$ such that
$\gcomplementof{\game'}{\cset_{n-1}(\game')}\;\subseteq\;\xwinset(\strat_1,\yfinitestrats(\game'))(\game',{\xparity}^{>0} )$.
We define $\xstrat_c(\state):=\strat_1(\state)$.
(Later, we will prove that in fact $\strat_1$ is winning \as)
\item
$\state\in \complementof{\xset_\tp}\cap\colorset{\zset_\tp}< n$.
Define $\xstrat_c(\state):=\xforcestrat(\game\cut\xset_\tp,\colorset{\zset_\tp}=n)(\state)$.

\item
$\state\in \complementof{\xset_\tp}\cap\colorset{\zset_\tp}=n$.
%
%By  Lemma~\ref{not:reach:trap:lemma} we know that 
Lemma~\ref{not:reach:trap:lemma} shows
$\postof\game\state\cap\complementof{\xset_\tp}\neq\emptyset$.
Define 
$\xstrat_c(\state):=\selectfrom{\postof\game\state\cap\complementof{\xset_\tp}}$.
\end{enumerate}

\parg{Correctness of $\xstrat_c$.}
Let $\ystrat\in\yfinitestrats(\game)$ be a finite-memory strategy for Player~$\py$.
We show that
$\probm_{\game,\state,\xstrat_c,\ystrat}(\xparity )=1$
for any state $\state\in\cset_n(\game)$.

First, we give a straightforward proof that any run 
$\state_0\state_1\cdots\in\runsof{\game,\state,\xstrat_c,\ystrat}$ will always stay inside
$\complementof{\xset_\tp}$, i.e.,
 $\state_i\in\complementof{\xset_\tp}$ for all $i\geq 0$.
We use induction on $i$.
The base case follows from $\state_0=\state\in\complementof{\xset_\tp}$.
For the induction step, we assume that
$\state_i\in\complementof{\xset_\tp}$, and show that
$\state_{i+1}\in\complementof{\xset_\tp}$.
We consider the following cases:
\begin{itemize}
\item
$\state_i\in\ycut{\complementof{\xset_\tp}}\cup\rcut{\complementof{\xset_\tp}}$.
The result follows
since $\complementof{\xset_\tp}$ is a
($\py$)-trap in $\game$
(by Lemma~\ref{not:reach:trap:lemma}).
\item 
$\state_i\in\xcut{\complementof{\xset_\tp}\cap\complementof{\zset_\tp}}$.
We know that $\state_{i+1}=\strat_1(\state_i)$.
Since $\strat_1\in\xnomemstrats(\game\cut\xset_\tp\cut\zset_\tp)$ it follows that
$\state_{i+1}\in\complementof{\xset_\tp}\cap\complementof{\zset_\tp}$, and in particular
$\state_{i+1}\in\complementof{\xset_\tp}$.
\item
$\state_i\in\xcut{\complementof{\xset_\tp}\cap\colorset{\zset_\tp}< n}$.
We know that $\state_{i+1}=\xforcestrat(\game\cut\xset_{\tp},\colorset{\zset_\tp}=n) (\state_i)$.
The result follows by the fact that
$\xforcestrat(\game\cut\xset_{\tp},\colorset{\zset_\tp}=n)$ is a strategy
 in $\game\cut\xset_\tp$.
\item
$\state_i\in\xcut{\complementof{\xset_\tp}\cap\colorset{\zset_\tp}=n}$.
We have $\state_{i+1}\in
\postof\game{\state_i}\cap\complementof{\xset_\tp}$,
and in particular $\state_{i+1}\in\complementof{\xset_\tp}$.
\end{itemize}
We now prove the main claim.
This is where we need the assumption of finite attractor and
finite-memory strategies.
Let us again consider a run 
$\run\in\runsof{\game,\state,\xstrat_c,\ystrat}$.
We show that $\run$ is \as\  winning for Player~$\px$
with respect to $\xparity$ in $\game$.
Let $\ystrat$ be induced by a memory structure
$\memstratn=\memstrattuple$.
Let $\tsys$ be the transition system induced
by $\game$, $\xstrat_c$, and $\ystrat$.
As explained in Section~\ref{prels:section},
$\run$ will \as\  visit
a configuration $\tuple{\state_\attractor,\memconf}\in\bscc$ 
for some BSCC $\bscc$ in $\tsys$.
Since there exists a finite attractor,
each state that occurs in $\bscc$ will \as\ 
be visited infinitely often by $\run$.
Let $\maxcolor$ be the maximal color occurring among
the states of $\bscc$.
There are two possible cases:
\begin{itemize}
\item $\maxcolor=n$.
  Since each state in $\game$ has color at most $n$, 
  Player~$\px$ will \as\  win.
\item $\maxcolor<n$.
  This implies that
  $\setcomp{\state_\bscc}{\tuple{\state_\bscc,\memconf}\in\bscc}
  \subseteq\complementof\zset_\tp$, and hence
  Player~$\px$ uses the strategy
  $\strat_1$ to win the game in $\game\cut\xset_\tp\cut\zset_\tp$ \wpp
  Then, either 
  (i) $\maxcolor\bmod 2=\px$ in which case
  all states inside $\bscc$ are almost sure winning for
  Player~$\px$; or
  (ii) $\maxcolor\bmod 2=\py$ in which case
  all states inside $\bscc$ are almost sure losing for Player~$\px$.
  The result follows from the fact that case (ii) gives a contradiction
  since all states in 
  $\game\cut\xset_\tp\cut\zset_\tp$
  (including those in $\bscc$) are winning for Player~$\px$ \wpp
\end{itemize}

%%%%%%%%%%%
\parg{Construction of $\ystrat_c$.}

We define a strategy $\ystrat_c$ such that, for all $\alpha$,
the following inclusion holds:
$\xset_\alpha\subseteq\yset_\alpha\subseteq\ywinset(\ystrat_c,\xfinitestrats(\game))(\game,{\yparity}^{>0})$.
The result then follows from the definition of $\cset_n(\game)$.
The inclusion $\xset_\alpha\subseteq\yset_\alpha$ holds by the definition of $\yset_\alpha$.
For any state $\state\in\gcomplementof\game{\cset_n(\game)}$,
we define $\ystrat_c(\state)$ as follows.
Let $\alpha$ be the smallest ordinal such that $s\in\yset_\alpha$.
Such an $\alpha$ exists by the well-ordering of ordinals
and since $\gcomplementof\game{\cset_n(\game)}=\bigcup_{\beta\in\ord}\xset_\beta=\bigcup_{\beta\in\ord}\yset_\beta$.
Now there are two cases:
\begin{itemize}
\item
  $\state\in\xset_\alpha\setminus\bigcup_{\beta<\alpha}\yset_\beta$.
  Define
  $\ystrat_c(\state):=
  f_1(\state):=
  \yforcestrat(\game,\bigcup_{\beta<\alpha}\yset_\beta)(\state)$.

\item
  $\state\in\cset_{n-1}(\game\cut\xset_\alpha\cut\zset_\alpha)$.
  By the induction hypothesis (on $n$),
  there is a memoryless strategy
  $\strat_2\in\ynomemstrats(\game\cut\xset_\alpha\cut\zset_\alpha)$ of Player~$\py$ such that
  $\state\in
  \ywinset(\strat_2,\xfinitestrats(\game\cut\xset_\alpha\cut\zset_\alpha))
  (\game\cut\xset_\alpha\cut\zset_\alpha,{\yparity}^{=1} )$.
  Define $\ystrat_c(\state):=\strat_2(\state)$.
\end{itemize}
\parg{Correctness of $\ystrat_c$.}
Let $\xstrat\in\xfinitestrats(\game)$ be a finite-memory strategy for
Player~$\px$.
We now use induction on $\alpha$ to show that
$\probm_{\game,\state,\ystrat_c,\xstrat}(\yparity )>0$
for any state $\state\in\yset_\alpha$.
There are three cases:\looseness=-1

\begin{enumerate}
\item
  If $\state\in\bigcup_{\beta<\alpha}\yset_\beta$,
  then $\state \in \yset_\beta$ for some $\beta < \alpha$
  and the result follows
  by the induction hypothesis on $\beta$.
\item
  If $\state\in\xset_\alpha \setminus \bigcup_{\beta<\alpha}\yset_\beta$, then we know that
  Player~$\py$ can use $\strat_1$ to force the game \wpp\  to
  $\bigcup_{\beta<\alpha}\yset_\beta$
  from which she wins \wpp
\item
  If $\state\in\cset_{n-1}(\game\cut\xset_\alpha\cut\zset_\alpha)$, then
  Player~$\py$ uses $\strat_2$.
  There are now two sub-cases:
  either (i) there is a run from $s$
  consistent with $\xstrat$ and $\ystrat_c$
  that reaches $\xset_\alpha$;
  or (ii) there is no such run.

  In sub-case (i), the run reaches $\xset_\alpha$ \wpp\ 
  Then, by cases~1 and~2, Player~$\py$ wins \wpp\looseness=-1

  In sub-case (ii), all runs stay forever outside $\xset_\alpha$.
  So the game is in effect played on $\game\cut\xset_\alpha$.
  Notice then that any run from $\state$ that is consistent
  with $\xstrat$ and $\ystrat_c$
  stays forever in $\game\cut\xset_\alpha\cut\zset_\alpha$.
  The reason is that (by Lemma~\ref{not:reach:trap:lemma})
  $\gcomplementof{\game\cut\xset_\alpha}\zset_\alpha$
  is an $\px$-trap in
  $\game\cut\xset_\alpha$.
  Since all runs remain inside $\game\cut\xset_\alpha\cut\zset_\alpha$,
  Player~$\py$ wins \wpp\  (even \as) wrt.\ $\yparity$ using $\strat_2$.
  \qedhere
\end{enumerate}

%%%%%%%%%%%%%%%

\end{proof}
The following theorem  follows immediately from the  previous lemmas.
\begin{thm}
  Stochastic parity games with almost sure winning conditions on
  infinite graphs are memoryless determined, provided there exists a
  finite attractor and the players are restricted to finite-memory
  strategies.
\end{thm}

\parg{Remark.}
We can compute both the \as\  winning set and the \wpp\  winning set
for both players as follows.
Let $\maxcolor$ be the maximal color occurring in the game.
Then:
\begin{itemize}
\item Player~$\px$ wins \as\  in $\cset_{\maxcolor}(\game)$ and \wpp\  in $\gcomplementof\game\cset_{\maxcolor+1}(\game)$;
\item Player~$\py$ wins \as\  in $\cset_{\maxcolor+1}(\game)$ and \wpp\  in $\gcomplementof\game\cset_{\maxcolor}(\game)$.
\end{itemize}

\section{Application to lossy channel systems}

\label{sec:application:sglcs}

\subsection{Lossy channel systems}
\label{sglcs:section}

A \emph{lossy channel system (LCS)}
is a finite-state machine
equipped with a finite number of unbounded fifo channels (queues) \cite{AbJo:lossy}.
The system is \emph{lossy} in the sense that, before and after a transition, an
arbitrary number of messages may be lost from the channels.
We consider \emph{stochastic game-LCS (SG-LCS)}: each individual %
message is lost independently with probability $\lossp$ 
in every step, where
$\lossp >0$ is a parameter of the system.
%
%In addition, 
The set of control states is 
partitioned into states belonging to Player~$\pz$ and~$\po$.
The player who owns the current control state chooses an enabled
outgoing transition. % to take.

Formally, a SG-LCS of rank $n$ is a tuple 
 $\sglcs=\sglcstuple$ where $\lcsstates$ is a
finite set of \emph{control states} partitioned into control states
$\lcsstatesz,\lcsstateso$ of Player~$\pz$ 
and~$\po$; % $\lcsstateso$ of \playerb;
$\channels$ is a finite set of \emph{channels}, $\msgs$ is a finite
set called the \emph{message alphabet}, $\lcstransitions$ is a set of
\emph{transitions}, $0<\lossp<1$ is the \emph{loss rate}, and
$\lcscoloring:\lcsstates\to\{0,\dots,n\}$ is the {\it coloring} function.
Each transition $\lcstransition\in\lcstransitions$ is of the form
$\lcsstate\transitionx{\op}\lcsstate'$, where
$\lcsstate,\lcsstate'\in\lcsstates$ and $\op$ is one of
the following three forms:
$\channel!\msg$ (send message $\msg\in\msgs$ in channel
$\channel\in\channels$), $\channel?\msg$ (receive message $\msg$ from
channel $\channel$), or $\nop$ (do not modify the channels).

The SG-LCS $\sglcs$ induces a game 
$\game=\gametuple$, where
$\states=\lcsstates\times(\msgs^*)^\channels\times\{0,1\}$.
That is, each state in the game (also called a \emph{configuration})
consists of a control state, a function
that assigns a finite word over the message alphabet to each channel,
and one of the symbols 0 or 1.
States where the last symbol is 0 are random:
$\rstates=\lcsstates\times(\msgs^*)^\channels\times\{0\}$.
The other states belong to a player according to the control state:
$\statesx=\lcsstatesx\times(\msgs^*)^\channels\times\{1\}$.
Transitions out of states of the form
$\state=(\lcsstate,\chassignment,1)$ model transitions in 
$\lcstransitions$ leaving control state $\lcsstate$.
On the other hand, transitions leaving configurations of the form
$\state=(\lcsstate,\chassignment,0)$ model message losses.
More precisely, transitions are defined as follows:\looseness=-1

\begin{itemize}
\item 
If $\state=(\lcsstate,\chassignment,1),
\state'=(\lcsstate',\chassignment',0)\in\states$, then 
we have $\state\transitionx{}\state'$ iff
$\lcsstate\transitionx{\op}\lcsstate'$ is a transition in $\lcstransitions$ and
(i)
  if $\op = \nop$, then $\chassignment=\chassignment'$;
(ii)
  if $\op = \channel!\msg$, then $\chassignment\channel = w$ and
  $\chassignment' = \chassignment[\channel \mapsto w \cdot \msg]$
%  $\chassignment'(\channel)=\chassignment(\channel)\msg$, and for all
%  $\channel'\in\channels-\{\channel\}$,
%  $\chassignment'(\channel')=\chassignment(\channel')$; and
(iii)
  if $\op = \channel?\msg$, then $\chassignment\channel = \msg\cdot w$ and
  $\chassignment' = \chassignment[\channel \mapsto w]$,
where the notation $\chassignment[\channel \mapsto w]$ represents the channel assignment which is the same as $\chassignment$
except that it maps $\channel$ to the word $w \in M^*$.
%  $\chassignment(\channel)=\msg\chassignment'(\channel)$, and for all
%  $\channel'\in\channels-\{\channel\}$,
%  $\chassignment'(\channel')=\chassignment(\channel')$.
%
%Then, there is a transition $\state\transition\state'$ in the game iff $\state\transitionx{\op}\state'$ for some operation $\op$.

\item 
To model message losses, we introduce the subword ordering $\preceq$
on words: $x\preceq y$ iff $x$ is a word obtained by removing zero or
more messages from arbitrary positions of $y$.
This is extended to channel contents
$\chassignment,\chassignment'\in(\msgs^*)^\channels$ by
$\chassignment\preceq\chassignment'$ iff
$\chassignment(\channel)\preceq\chassignment'(\channel)$ for all
channels $\channel\in\channels$, and to configurations
$\state=(\lcsstate,\chassignment,i),\state'=(\lcsstate',\chassignment',i')\in\states$
by $\state\preceq\state'$ iff $\lcsstate=\lcsstate'$,
$\chassignment\preceq\chassignment'$, and $i=i'$.
For any $\state=(\lcsstate,\chassignment,0)$ and any $\chassignment'\preceq\chassignment$, there is a transition
$\state\transition(\lcsstate,\chassignment',1)$.
The probability of random transitions is given by
$\probp((\lcsstate,\chassignment,0),(\lcsstate,\chassignment',1)) =
a\cdot\lossp^{c-b}\cdot(1-\lossp)^c$, where $a$ is the number of ways to
obtain $\chassignment'$ by losing messages in $\chassignment$, $b$ is
the total number of messages in all channels of $\chassignment$, and $c$ is the
total number of messages in all channels of $\chassignment'$
(see \cite{Parosh:etal:attractor:IC} for details).\looseness=-1
\end{itemize}

\noindent
Every configuration of the form $(\lcsstate,\chassignment,0)$ has at least one
successor, namely $(\lcsstate,\chassignment,1)$.
If a configuration $(\lcsstate,\chassignment,1)$ does not have successors
according to the rules above, then we add a transition
$(\lcsstate,\chassignment,1)\transition(\lcsstate,\chassignment,0)$,
to ensure that the induced game is sink-free.
%(this is not a restriction; only a technical detail in order
%to make the exposition go through smoothly).

Finally, for a configuration $\state=(\lcsstate,\chassignment,i)$, we define
$\coloring(\state):=\lcscoloring(\lcsstate)$.
Notice that the graph of the game is bipartite, in the sense that
a configuration in $\rstates$ has only
transitions to configurations in $\zocut\states$,
and vice versa.

We say that a set of channel contents $\tt X \subseteq (\msgs^*)^\channels$
is \emph{regular} if it is a finite union of sets of the form $\tt Y \subseteq (\msgs^*)^\channels$
where $\tt Y(c)$ is a regular subset of $\msgs^*$ for every $c \in \channels$
(this coincides with the notion of recognisable subset of $(\msgs^*)^\channels$;
cf. \cite{Berstel}).
We extend the notion of regularity to a set of configurations $P \subseteq \states$
by saying that $P$ is \emph{regular} iff, for every control state $\lcsstate \in \lcsstates$ and $i \in \{0,1\}$,
there exists a regular set of channel contents $\tt X_{\lcsstate, i} \subseteq (\msgs^*)^\channels$
s.t. $P = \setcomp {(\lcsstate, \chassignment, i)} {\lcsstate \in \lcsstates, i \in \{0,1\}, \chassignment \in \tt X_{\lcsstate, i}}$.

In the qualitative {\it parity game problem} for SG-LCS, we 
want to characterize the sets of configurations
where Player~$\px$ can force the \xparity{} condition to hold \as,
for both players.

\subsection{From scheme to algorithm}
\label{algorithm:section}
We transform the scheme of Section~\ref{parity:section}
into an algorithm for deciding the \as\ parity game problem for SG-LCS.
Consider an SG-LCS $\sglcs=\sglcstuple$ and
the induced game $\game=\gametuple$
of some rank $n$.
Furthermore,  assume that the players
are restricted to finite-memory strategies.
We show the following.
\begin{thm}
\label{thm:memoryless:determinacy}
The sets of winning configurations for Players~$\pz$~and~$\po$
are effectively computable as regular sets of configurations.
Furthermore, from each configuration, memoryless strategies suffice for the winning player.
\end{thm}
In the statement of the theorem, ``effectively'' means that
a finite description of the regular sets of winning configurations is computable.
We give the proof in several steps.
First, we show that the game induced by an SG-LCS contains a finite
attractor (Lemma~\ref{fattractors:lemma}).
Then, we show that the scheme in
Section~\ref{reachability:section} for computing winning configurations
wrt.\ reachability objectives
is guaranteed to terminate (Lemma~\ref{reachable:termination:lemma}).
Furthermore, we show that the scheme in
Section~\ref{parity:section} for computing winning configurations
wrt. \as\  parity objectives is guaranteed to terminate
(Lemma~\ref{parity:termination:lemma}).
Notice that 
Lemmas~\ref{reachable:termination:lemma}~and~\ref{parity:termination:lemma}
imply that for SG-LCS our transfinite constructions 
stabilize below $\omega$ (the first infinite ordinal).
Finally, we show that each step in the above two schemes 
can be performed using standard operations on regular languages
(Lemmas~\ref{reachability:compuability:lemma}~and~\ref{parity:compuability:lemma}).

\parg{Finite attractor.}
In \cite{Parosh:etal:attractor:IC} it was shown that
any Markov chain induced
by a Probabilistic LCS contains a finite attractor.
The proof can be carried over in a straightforward manner
to the current setting.
More precisely, the finite attractor is given by
$\attractor=(\lcsstates\times\emptychannels\times\{0,1\})$
where $\emptychannels(\channel)=\emptyword$ for each $\channel\in\channels$.
In other words, $\attractor$ is given by
the set of configurations in which all channels are empty.
  The proof relies on the observation that if the number of messages
  in some channel is sufficiently large, it is more likely that the number of
  messages decreases than that it increases in the next step.
This gives the following.
\begin{lem}
\label{fattractors:lemma}
$\game$ contains a finite attractor.
\end{lem}

\parg{Termination of Reachability Scheme.}

For a set of configurations $\stateset\subseteq\states$,
we  define the {\it upward closure} of $\stateset$ by
$\ucof\stateset:=\setcomp{\state}{\exists\state'\in\stateset.\,\state'\preceq\state}$.
A set $\ucset\subseteq\stateset\subseteq\states$
is said to be {\it $\stateset$-upward-closed}
(or {\it $\stateset$-u.c.} for short) if
$(\ucof{\ucset})\cap\stateset=\ucset$.
We say that $\ucset$ is {\it upward closed} if it is
$\states$-u.c.

\newcommand{\tpf}{j}

\begin{lem}
\label{higman:lemma}
  If $\stateset_0\subseteq \stateset_1\subseteq\cdots$, 
and for all $i$ it holds that
  $\stateset_i\subseteq \stateset$ and $\stateset_i$ is $\stateset$-u.c., then there is an $\tpf\in\nat$
  such that $\stateset_i=\stateset_\tpf$ for all $i\geq \tpf$.
\end{lem}
\begin{proof}
  By Higman's lemma \cite{higman:divisibility}, there is a $\tpf\in\nat$ s.t.
  $\stateset_i\!\uparrow=\stateset_\tpf\!\uparrow$ for all $i\geq \tpf$.
  Hence, $\stateset_i\!\uparrow\cap \stateset=\stateset_\tpf\!\uparrow\cap \stateset$
  for all $i\geq\tpf$.
  Since all $\stateset_i$ are $\stateset$-u.c.,
  $\stateset_i\!\uparrow\cap \stateset=\stateset_i$
  for all $i\geq\tpf$.
  So $\stateset_i=\stateset_\tpf$ for all $i\geq \tpf$.
\end{proof}

Now, we can show termination of the reachability scheme.
\begin{lem}
\label{reachable:termination:lemma}
  There exists a finite $\tpf\in\nat$ such
  that $\reachset_i=\reachset_\tpf$ for all $i\geq\tpf$.
\end{lem}
\begin{proof}
First, we show that $\rcut{\reachset_i\setminus\targetset}$ is 
$(\gcomplementof\game\targetset)$-u.c. for all $i\in\nat$.
We use induction on $i$.
For $i=0$ the result is trivial since
$\reachset_i\setminus\targetset=\emptyset$.
For $i>0$, suppose that
$\state=(\lcsstate,\chassignment,0)\in\rcut{\reachset_i}\setminus\targetset$.
This means that 
$\state\transition(\lcsstate,\chassignment',1)\in\reachset_{i-1}$
for some $\chassignment'\preceq\chassignment$, and hence
$\state'\transition(\lcsstate,\chassignment',1)$
for all $\state'$ s.t. $\state\preceq\state'$.

By Lemma~\ref{higman:lemma},
there is a $\tpf'\in\nat$
  such that $\rcut{\reachset_i}\setminus\targetset=\rcut{\reachset_{\tpf'}}\setminus\targetset$ 
 for all $i\geq \tpf'$.
Since $\reachset_i\supseteq\targetset$ for all $i\geq 0$ it follows that 
%there is a $\tp'\in\nat$ such that
 $\rcut{\reachset_i}=\rcut{\reachset_{\tpf'}}$
 for all $i\geq \tpf'$.

Since the graph of $\game$ is bipartite
(as explained in Section~\ref{sglcs:section}),
$\xcut{\preof\game{\reachset_i}}=\xcut{\preof\game{\rcut{\reachset_i}}}$
and
$\ycut{\dualpreof\game{\reachset_i}}=\ycut{\dualpreof\game{\rcut{\reachset_i}}}$.
Since $\rcut{\reachset_i}=\rcut{\reachset_{\tpf'}}$ for all $i\geq\tpf'$,
we have
$\xcut{\preof\game{\reachset_i}}=\xcut{\preof\game{\rcut\reachset_{\tpf'}}}\subseteq\reachset_{\tpf'+1}$ and
$\ycut{\dualpreof\game{\reachset_i}}=\ycut{\dualpreof\game{\rcut\reachset_{\tpf'}}}\subseteq\reachset_{\tpf'+1}$.
It then follows that $\reachset_i=\reachset_{\tpf}$ for all $i\geq\tpf:=\tpf'+1$.
\end{proof}

\parg{Termination of Parity Scheme.}

We prove that the scheme from Section~\ref{parity:section}
terminates under the condition that the reachability sets are computable
and that there exists a finite attractor.
This suffices since, by the part above, the reachability scheme terminates,
thus yielding computability of the reachability set.
However, here we prove termination of the parity scheme with no further assumption on the reachability sets other than their computability. 

We first prove two immediate auxiliary lemmas.
\begin{lem}
  \label{closable:attractor:lemma}
  A closable set intersects every attractor.
\end{lem}
\begin{proof}
  In any closable set, the players can choose strategies that force
  the game to remain in the set surely.
  The lemma now follows since an attractor is visited almost surely
  by any run, and this would be impossible if the attractor did not
  have any element in the set.
\end{proof}

\begin{lem}
\label{cset:trap:lemma}
$\cset_n(\game)$ is a $(\py)$-trap.
\end{lem}
\begin{proof}
$\cset_0(\game)$ is trivially a $(\py)$-trap.
For  $i\geq 1$, the result follows immediately from the definition of $\cset_n(\game)$ in Eq~\ref{eq:cset:def}
as the complement of a force set (by Lemma~\ref{not:reach:trap:lemma}).
\end{proof}

\begin{lem}
\label{parity:termination:lemma}
There is a finite $\tpf\in\nat$ such that $\xset_i=\xset_\tpf$ for all
$i\geq\tpf$.
\end{lem}
\begin{proof}
  We will prove the claim by showing that
  $\cset_{n-1}(\game\cut\xset_i\cut\zset_i)$ in the definition of
  $\yset_i$ contains an element from the attractor, and that the
  $\cset_{n-1}(\game\cut\xset_i\cut\zset_i)$ sets constructed in
  different steps $i$ are disjoint.
  First, $\cset_{n-1}(\game\cut\xset_i\cut\zset_i)$ is an $\px$-trap
  by Lemma~\ref{cset:trap:lemma}.
  Hence it is closable, and therefore
  Lemma~\ref{closable:attractor:lemma} implies that it contains an
  element from the attractor.
  Second, by the definition of the $\cut$ operator, $\xset_i$ and
  $\game\cut\xset_i\cut\zset_i$ are disjoint.
  Since
  $\cset_{n-1}(\game\cut\xset_i\cut\zset_i)\subseteq S\setminus\xset_i\setminus\zset_i$,
  it follows that $\yset_i$ is the \emph{disjoint} union of $\xset_i$ and
  $\cset_{n-1}(\game\cut\xset_i\cut\zset_i)$.
  Then, the definition of $\xset_i$ implies that
  $\cset_{n-1}(\game\cut\xset_i\cut\zset_i)\subseteq\yset_i\setminus\bigcup_{j<i}\yset_j$.
  Hence, if $j\neq i$, $\cset_{n-1}(\game\cut\xset_i\cut\zset_i)$ and
  $\cset_{n-1}(\game\cut\xset_j\cut\zset_j)$ are disjoint.
  Since all $\cset_{n-1}(\game\cut\xset_i\cut\zset_i)$ sets are
  disjoint, and each of them contains at least one element of the
  attractor, and the attractor is finite, the algorithm terminates in
  at most $|\attractor|$ steps.
\end{proof}

\parg{Computability.}
Regular languages of configurations are effectively closed under the operations of
upward-closure, predecessor, set-theoretic union, intersection, and complement \cite{AbdullaBouajjaniDorso:JLC:2008}.
For completeness, we show these properties below.
\begin{lem}
	\label{lem:upward-closure:regular}
	If $P$ is a regular set of configurations, then its upward-closure $\ucof P$ is effectively regular.
\end{lem}
\begin{proof}
	A regular set $P$ of configurations is by definition of the form
        \[
	P = \setcomp {(\lcsstate, \chassignment, i)} {\lcsstate \in \lcsstates, i \in \{0,1\}, \chassignment \in \tt X_{\lcsstate, i}}
        \]
	where the $\tt X_{\lcsstate, i}$'s are regular sets of channel contents.
	It thus suffices to show that
	$\ucof {\tt X} := \setcomp{\chassignment}{\exists\chassignment'\in\tt X.\,\chassignment'\preceq\chassignment}$
	is an effectively regular set of channel contents
	when $\tt X$ is a regular set of channel contents.
	By definition, $\tt X$ is a finite union of sets of the form
	$\tt Y \subseteq (\msgs^*)^\channels$ with $\tt Y(c)$ regular for every $c \in \channels$.
	Let $\ucof {\tt X}$ be the union of the $\ucof {\tt Y}$,
	where, for every  $c \in \channels$,
	a finite automaton recognizing $\ucof {\tt Y}(c)$ is obtained from a finite automaton recognizing $\tt Y(c)$
	by adding a self-loop labeled with $M$ on every state thereof.
\end{proof}

\begin{lem}
	\label{lem:boolean_op:regular}
	If $P, Q$ are regular sets of configurations,
	then $P \cup Q$, $P \cap Q$, and $\states \setminus P$
	are effectively regular sets of configurations.
\end{lem}
\begin{proof}
	The proof is very similar to the one in the previous lemma,
	by exploiting the fact that regular languages are closed under the operations of union, intersection, and complement.
\end{proof}

\begin{lem}
	\label{lem:pre:regular}
	If $P$ is a regular set of configurations,
	then $\preof\game P$ is an effectively regular set of configurations.
\end{lem}
\begin{proof}
	Let $P$ be a regular set of configurations.
	By a case analysis on which transition is taken, we can write
	\begin{align*}
		\preof\game P = \bigcup_{\lcstransition \in \lcstransitions} \preof\game {P, \lcstransition} \cup \preRof\game P
	\end{align*}
	where 
	\begin{align*}
		\preof\game {P, \lcsstate\transitionx\nop\lcsstate'}
			&:= \setcomp {(\lcsstate, \chassignment, 1)} {(\lcsstate', \chassignment, 0) \in P} \\
		\preof\game {P, \lcsstate\transitionx{\channel!\msg}\lcsstate'}
			&:= \setcomp {(\lcsstate, \chassignment, 1)} {(\lcsstate', \chassignment', 0) \in P.\,
			\chassignment'(c) = w \cdot \msg,
			\chassignment = \chassignment[\channel \mapsto w]} \\
		\preof\game {P, \lcsstate\transitionx{\channel?\msg}\lcsstate'}
			&:= \setcomp {(\lcsstate, \chassignment, 1)} {(\lcsstate', \chassignment', 0) \in P.\,
			\chassignment = \chassignment'[\channel \mapsto \msg \cdot \chassignment(\channel)]} \\
		\preRof\game P
			&:= \setcomp {(\lcsstate, \chassignment, 0)} {(\lcsstate', \chassignment', 1) \in P.\,
			\chassignment' \preceq \chassignment} = \ucof {\setcomp {(\lcsstate, \chassignment', 0)} {(\lcsstate', \chassignment', 1) \in P}}
	\end{align*}
	Then, $\preof\game {P, \lcsstate\transitionx\nop\lcsstate'}$ is clearly effectively regular,
	$\preof\game {P, \lcsstate\transitionx{\channel!\msg}\lcsstate'}$ is regular,
	because regular languages are effectively closed under (right) quotients,
	$\preof\game {P, \lcsstate\transitionx{\channel?\msg}\lcsstate'}$ is regular,
	because regular language are effectively closed under (left) concatenation with single symbols,
	and $\preRof\game P$ is effectively regular by Lemma~\ref{lem:upward-closure:regular}.
\end{proof}

The lemmas above show that all operations
used in computing $\xforceset(\game,\targetset)$ effectively preserve regularity.
Thus we obtain the following lemma.

\begin{lem}
\label{reachability:compuability:lemma}
If $\targetset$ is regular, then
$\xforceset(\game,\targetset)$ is effectively regular.
\end{lem}
\begin{lem}
\label{parity:compuability:lemma}
For each $n$, $\cset_n(\game)$ is
effectively regular.
\end{lem}
\begin{proof}
The set $\states$ is regular, and hence 
$\cset_0(\game)=\states$
is effectively regular.
The result for $n>0$ follows from
Lemma~\ref{reachability:compuability:lemma} and from the fact
that the rest of the operations used to build  $\cset_n(\game)$
are those of set complement and union.
\end{proof}

\subsection{Construction of regular winning strategies}

\label{sec:strategy:construction}

In this section,
we show that the memoryless winning strategies constructed in Theorem~\ref{thm:memoryless:determinacy}
can be finitely represented as a (finite) list of rules with regular guards on the channel contents.
This representation can be easily turned in a more low-level one, e.g.,
a finite automaton with output reading the channel contents and outputting the rule do be played next,
but for the ease of presentation we have chosen a more high-level description.

\parg{Preliminaries.}

Let $\sglcs=\sglcstuple$ be a SG-LCS.
A \emph{(memoryless) regular SG-LCS strategy} $\lcsstrat$ for Player~$\px$ is a finite list of guarded rules
$\{ \lcsstate_i, X_i \transitionx{\op_i} \lcsstate'_i \}_{i=1}^n$,
where the \emph{guard} $X_i\subseteq (\msgs^*)^\channels$ is a regular
set of channel contents and
$\lcsstate_i\transitionx{\op_i}\lcsstate'_i$ is a transition in
$\lcstransitions$ s.t. $\lcsstate_i\in\lcsstatesx$ and:
\begin{itemize}
\item If $\op_i = \channel?\msg$, every $\chassignment\in X_i$ has $\msg$ as the first symbol of $\chassignment(\channel)$.
\item Guards for the same control state are disjoint; i.e., for each
  $i, j$, if $\lcsstate_i=\lcsstate_j$ then $X_i\cap X_j=\emptyset$.
\end{itemize}

\noindent
The \emph{domain} of a regular SG-LCS strategy $\lcsstrat$ is
\begin{align*}
	\domof{\lcsstrat} = \setcomp{(\lcsstate, \chassignment)}
	{\textrm{ there exists a guarded rule } \lcsstate, X \transitionx{\op} \lcsstate' \in \lcsstrat \textrm{ s.t. } \chassignment \in X}
\end{align*}
Intuitively, the rule $(\lcsstate_i,X_i\transitionx{\op_i}\lcsstate'_i)$ should be applied from control state $\lcsstate_i$ if the channel contents belong to the guard $X_i$.
Formally, let $\game=\gametuple$ be the game induced by $\sglcs$.
The (partial, memoryless) \emph{induced strategy} $\induced\lcsstrat$ of a regular SG-LCS strategy $\lcsstrat$ is defined,
for every $(\lcsstate,\chassignment)\in\domof{\lcsstrat}$,
as $\induced\lcsstrat(\lcsstate, \chassignment, 1) = (\lcsstate'_i, \chassignment', 0)$,
where $\lcsstate_i, X_i \transitionx{\op} \lcsstate'_i$ is the unique guarded rule in $\lcsstrat$ such that $\lcsstate_i=\lcsstate$ and $\chassignment \in X_i$,
and $\chassignment'$ is the unique channel contents s.t.
$(\lcsstate, \chassignment, 1) \transition (\lcsstate', \chassignment', 0)$ in the game $\game$.
If $(\lcsstate,\chassignment)\not\in\domof{\lcsstrat}$, then $\induced\lcsstrat(\lcsstate,\chassignment,1)=\bot$.

Given two regular SG-LCS strategies $\lcsstrat_0, \lcsstrat_1$ with disjoint domains,
their \emph{union} $\lcsstrat_0 \cup \lcsstrat_1$ is the regular SG-LCS strategy obtained by concatenating the lists of guarded rules of $\lcsstrat_0$ and $\lcsstrat_1$.

Given two sets of configurations $Q, Q' \subseteq \states$,
a $\emph{selection function}$ from $Q$ to $Q'$ is any function $f : Q \mapsto Q'$ s.t.,
for every $(\lcsstate, \chassignment) \in Q$,
\begin{align*}
	 f(\lcsstate, \chassignment) \in \left(\postof\game{\lcsstate, \chassignment}\cap Q'\right)
\end{align*}
In other words, a selection function picks a legal successor in $Q'$ for every configuration in $Q$.

\parg{Construction.}

The rest of this section is devoted to the construction of regular winning strategies for both players,
as summarised by the following theorem.
\begin{thm}
\label{thm:regular:strategies}
Memoryless winning strategies for both players are effectively computable as regular SG-LCS strategies.
\end{thm}

We begin by showing that, if the set of selection functions is non-empty,
then there are simple selection functions induced by \emph{regular} SG-LCS strategies.
\begin{lem}
	\label{lemma:select:regular}
	Let $Q, Q' \subseteq \states$ be two regular sets of configurations.
	If there exists a selection function from $Q$ to $Q'$,
	then there exists a regular SG-LCS strategy $\lcsstrat$
	s.t. $\induced\lcsstrat$ is a selection function from $Q$ to $Q'$.\looseness=-1
\end{lem}
\begin{proof}
	Let $f$ be a selection function from $Q$ to $Q'$;
	in particular, the set $\postof\game{\lcsstate, \chassignment}\cap Q'$ is non-empty
	for each $(\lcsstate, \chassignment) \in Q$.
	Let $\lcstransitions = \set{\lcsstate_0 \transitionx{\op_0} \lcsstate_0', \dots, \lcsstate_k \transitionx{\op_k} \lcsstate_k'}$
	be the finitely many transitions of $\sglcs$.
	For every $i \in \set{0, \dots, k}$,
	let $P_i$ be the set of predecessors of $Q'$ in $Q$ via transition $\lcsstate_i \transitionx{\op_i} \lcsstate_i'$, i.e.,
	\begin{align*}
		P_i = \preof\game {Q', \lcsstate_i \transitionx{\op_i} \lcsstate_i'} \cap Q
		= \setcomp {(\lcsstate_i, \chassignment) \in Q} { \textrm{there exists } (\lcsstate_i', \chassignment') \in Q' \cdot (\lcsstate_i, \chassignment) \transitionx{\op_i} (\lcsstate_i', \chassignment') }
	\end{align*}
	Since $Q, Q'$ are regular,
	$\preof\game {Q', \lcsstate_i \transitionx{\op_i} \lcsstate_i'}$ is regular (cf. Lemma~\ref{lem:pre:regular}),
	and thus $P_i$ is regular too.
	Consider the sequence of (regular) sets $Q_0 = P_0$, and, for $0 < i \leq k$, $Q_i = P_i \setminus \bigcup_{0 \leq j < i} Q_j$,
	and let $Q_{i_0}, \dots, Q_{i_h}$ be the subsequence of non-empty sets.
	Then, $\set{Q_{i_0}, \dots, Q_{i_h}}$ is a (regular) partition of $Q$:
	The sets are disjoint by definition, and each $(\lcsstate, \chassignment) \in Q$
	belongs to some $Q_{i_j}$ since $\postof\game{\lcsstate, \chassignment}\cap Q'$ is non-empty.
	Let $\set{X_{i_0}, \dots, X_{i_h}} \subseteq 2^{(\msgs^*)^\channels}$ be the set of regular channel contents
	s.t., for $0 \leq j \leq h$,
	$Q_{i_j}$ is of the form $\setcomp{(\lcsstate_{i_j}, \chassignment)}{\chassignment \in X_{i_j}}$.
	Let $\lcsstrat$ be the following regular SG-LCS strategy:
	\begin{align}
	  \{ s_{i_0}, X_{i_0} \transitionx{\op_{i_0}} \lcsstate'_{i_0}, \dots, s_{i_h}, X_{i_h} \transitionx{\op_{i_h}} \lcsstate'_{i_h} \}
	\end{align}
	By definition, $\induced\lcsstrat$ is a selection function from $Q$ to $Q'$.
\end{proof}

In the next lemma, we show that regular SG-LCS strategies suffice to keep the game in regular traps.
\begin{lem}
	\label{lemma:trap:certainly:regular}
	If $\stateset$ is a $(\py)$-trap and regular, then there exists
	a regular SG-LCS strategy $\lcsstrat$ for Player~$\px$ 
	such that
	$\stateset\subseteq\xvinset(\induced\lcsstrat,\yallstrats(\game))(\game,\always\stateset)$.
\end{lem}
\begin{proof}
	By Lemma~\ref{trap:certainly:lemma},
	there exists a memoryless strategy $\xstrat$ for Player~$\px$ with the required property.
	Moreover, by inspecting the proof of the lemma,
	we can see that $\xstrat$ is defined as
	$\xstrat(\state)=\selectfrom{\postof\game{\state}\cap\stateset}$
	for every configuration $\state \in \xcut\stateset$,
	i.e., $\xstrat$ is a selection function from $\xcut\stateset$ to $\stateset$,
	and, in fact, any such selection function can be taken.
	By Lemma~\ref{lemma:select:regular},
	there exists a regular SG-LCS strategy $\lcsstrat$ s.t.
	the induced strategy $\induced \lcsstrat$ is a selection function from $\xcut\stateset$ to $\stateset$.
\end{proof}

The following lemma shows that there are regular SG-LCS strategies for the reachability and safety objective
(cf. Lemma~\ref{reachability:correct:lemma}).
\begin{lem}
	\label{lemma:reachability:regular:strategy}
	Let $\targetset \subseteq \states$ be a regular set of configurations.
	There exist regular SG-LCS strategies $\xlcsforcestrat(\game,\targetset)$ for Player~$\px$
	and $\ylcsavoidstrat(\game,\targetset)$ for Player~$\py$~s.t.
	%inducing strategies $\xforcestrat(\game,\targetset)$ and $\yavoidstrat(\game,\targetset)$, respectively.
	\begin{align*}
	 \xforceset(\game,\targetset)&\subseteq
		\xwinset(\induced{\xlcsforcestrat(\game,\targetset)},\yallstrats(\game))(\game,\eventually\targetset^{>0}) \\
	 \yavoidset(\game,\targetset)&\subseteq
		\yvinset(\induced{\ylcsavoidstrat(\game,\targetset)},\xallstrats(\game))(\game,\always(\gcomplementof\game\targetset))
	\end{align*}
\end{lem}
\begin{proof}
	We first show a regular SG-LCS strategy for Player~$\px$ for the reachability objective.
	Consider the sequence of sets $\reachset_0, \reachset_1, \dots$ constructed in Section~\ref{reachability:section}.
	By Lemma~\ref{reachable:termination:lemma}, there exists $j \in \nat$ s.t. $\forall i > j$, $\reachset_i = \reachset_j$.
	Moreover, since $\reachset_i$ is built starting from the regular set $\targetset$
	and according to regularity-preserving operations
	(union, predecessor, and complement; cf. Lemmas~\ref{lem:boolean_op:regular} and \ref{lem:pre:regular}),
	$\reachset_i$ is regular for every $0 \leq i \leq j$.
	Consider the sequence of regular sets $R_0 = \reachset_0$ and $R_i = \reachset_i \setminus \reachset_{i-1}$
	for every $0 < i \leq j$.
	Recall the definition of $\xforcestrat(\game,\targetset)$ in the proof of Lemma~\ref{reachability:correct:lemma}:
	For every $0<i\leq j$,
	$\xforcestrat(\game,\targetset)$ was uniformly defined on $R_i$ as
	\begin{align*}
		\xforcestrat(\game,\targetset)(s) = \selectfrom{\postof\game{\state}\cap\reachset_{i - 1}}.
	\end{align*}
	Therefore, there exists a selection function from $R_i$ to $\reachset_{i-1}$, for every $0 < i \leq j$.
	Since the $R_i$'s and $\reachset_i$'s are regular, by Lemma~\ref{lemma:select:regular},
	there exists a regular SG-LCS strategy $f_i$ with domain $R_i$ inducing such a selection function.
	Since the $R_i$'s are disjoint, and since any selection function is correct,
	take as $\xlcsforcestrat(\game,\targetset)$
	the union strategy $f_0 \cup \cdots \cup f_j$.
	Since the actual choice of the selection function is irrelevant, we conclude that
        \[
	\xforceset(\game,\targetset)\subseteq
	\xwinset(\induced{\xlcsforcestrat(\game,\targetset)},\yallstrats(\game))(\game,\eventually\targetset^{>0})
        \]
	We conclude the proof by providing the required regular SG-LCS strategy for Player~$\py$ for the safety objective.
	By Lemma~\ref{not:reach:trap:lemma}, $\yavoidset(\game,\targetset)$ is an $\px$-trap.
	Since $\yavoidset(\game,\targetset)$ is regular,
	by Lemma~\ref{lemma:trap:certainly:regular} there exists a regular SG-LCS strategy $\ylcsavoidstrat(\game,\targetset)$
	such that $\yavoidset(\game,\targetset)\subseteq
		\yvinset(\induced{\ylcsavoidstrat(\game,\targetset)},\xallstrats(\game))(\game,\always(\gcomplementof\game\targetset))$.
\end{proof}

To conclude the proof of Theorem~\ref{thm:regular:strategies},
we show that regular SG-LCS strategies suffice for the parity objective
(cf. Lemma~\ref{cn:infinite:termination:lemma}).
\begin{lem}
	\label{lemma:parity:regular:strategy}
	There are regular SG-LCS strategies
	$\xlcsstrat_c$ for Player~$\px$ and	$\ylcsstrat_c$ for Player~$\py$	such that 
	\begin{align*}
		\cset_n(\game)\;&\subseteq\;\xwinset(\xinducedlcsstrat_c,\yfinitestrats(\game))(\game,{\xparity}^{=1} ) \\
		\gcomplementof\game{\cset_n(\game)}\;&\subseteq\;\ywinset(\yinducedlcsstrat_c,\xfinitestrats(\game))(\game,{\yparity}^{>0} )
	\end{align*}
\end{lem}

\begin{proof}
	We define regular SG-LCS strategies $\xlcsstrat_c$ for Player~$\px$ and $\ylcsstrat_c$ for Player~$\py$
	by induction on $n \geq 1$.
	By inspecting the proof of Lemma~\ref{cn:infinite:termination:lemma},
	we note that winning strategies for both players are constructed according to a case analysis on disjoint regular domains,
	for which winning regular SG-LCS strategies exist either by induction hypothesis,
	or by Lemma~\ref{lemma:reachability:regular:strategy} (for reachability).
	Recall that, by Lemma~\ref{parity:termination:lemma},
	there exists $i \in \nat$ s.t. $\xset_j = \xset_i$ for every $j > i$.
	Moreover, all the sets $\xset_j, \yset_j, \zset_j$ involved in the construction are regular for every $0 \leq j \leq i$
	since they are constructed starting from regular sets and according to regularity-preserving operations
	(boolean operations, cf. Lemma~\ref{lem:boolean_op:regular};
	force-sets, cf. Lemma~\ref{lemma:reachability:regular:strategy}).

	\parg{Construction of $\xlcsstrat_c$.}
	Define the two regular sets of configurations $\complementof{\xset_j}:=\gcomplementof\game{\xset_j}$ and 
	$\complementof{\zset_j}:=\gcomplementof\game{\zset_j}$.
	By definition, $\cset_n(\game)=\complementof{\xset_j}$.
	Following Lemma~\ref{cn:infinite:termination:lemma},
	we define $\xlcsstrat_c(\state)$ depending on the membership of $\state$
	in one of the following three partitions of $\complementof{\xset_j}$:
	\begin{align*}
		\set{\complementof{\xset_j}\cap\complementof{\zset_j},
			\ \ \complementof{\xset_j}\cap\colorset{\zset_j}<n,
			\ \  \complementof{\xset_j}\cap\colorset{\zset_j}=n}
	\end{align*}
	In the first case, note that $\game\cut\xset_j\cut\zset_j$ does not contain any configurations of color $\geq n$ (cf.\break Lemma~\ref{cn:infinite:termination:lemma}).
	Thus, by the induction hypothesis, there is a regular SG-LCS strategy $\lcsstrat_j$
	for Player~$\px$ in $\game\cut\xset_j\cut\zset_j$ such that the
        induced strategy has domain $\complementof{\xset_j}\cap\complementof{\zset_j}$.
	%s.t. $\gcomplementof{\game'}{\cset_{n-1}(\game')}\;\subseteq\;\xwinset(\induced\lcsstrat_1,\yfinitestrats(\game'))(\game',{\xparity}^{>0}$.
	%refine to claim probability = 1 in fact?
	%
	In the second case, let $\lcsstrat_2$ be the regular SG-LCS strategy
	$\xlcsforcestrat(\game\cut\xset_j,\colorset{\zset_j}=n)$, for which
        the induced strategy has domain $\complementof{\xset_j}\cap\colorset{\zset_j}<n$
	(it exists by Lemma~\ref{lemma:reachability:regular:strategy}).
	Finally, in the third case, the strategy $\selectfrom{\postof\game\cdot\cap\complementof{\xset_j}}$
	witnesses the existence of a selection function
	from $\complementof{\xset_j}\cap\colorset{\zset_j}=n$ to $\complementof{\xset_j}$.
	Let $\lcsstrat_3$ be a regular SG-LCS strategy
	inducing a selection function from
	$\complementof{\xset_j}\cap\colorset{\zset_j}=n$ to
        $\complementof{\xset_j}$
	(it exists by Lemma~\ref{lemma:select:regular}).
	Then, $\xlcsstrat_c$ is defined as the union of the three previously constructed strategies:\looseness=-1
	\begin{align*}
		\xlcsstrat_c := \lcsstrat_1 \cup \lcsstrat_2 \cup \lcsstrat_3
	\end{align*}
	Since the actual choice of selection function is irrelevant,
	$\xlcsstrat_c$ induces a correct strategy by the same arguments as in the proof of Lemma~\ref{cn:infinite:termination:lemma},
	i.e., $\cset_n(\game)\;\subseteq\;\xwinset(\xinducedlcsstrat_c,\yfinitestrats(\game))(\game,{\xparity}^{=1})$.

	\parg{Construction of $\ylcsstrat_c$.}
	Recall that $\gcomplementof\game{\cset_n(\game)} = \yset_j = \xset_j$, and,
	for every $0 \leq i \leq j$,
	$\yset_i = \xset_i\cup\cset_{n-1}(\game\cut\xset_i\cut\zset_i)$.
	For every $1 \leq i \leq j$, let $\lcsstrat_i^1$ be the regular SG-LCS strategy $\ylcsforcestrat(\game,\yset_{i-1})$
	with domain $\xset_i \setminus \yset_{i-1}$ % (with the convention $\yset_{-1} = \emptyset$)
	(it exists by Lemma~\ref{lemma:reachability:regular:strategy}).
        By the induction hypothesis, there is also a regular SG-LCS strategy $\lcsstrat_i^2$ such that the induced strategy has domain $\cset_{n-1}(\game\cut\xset_i\cut\zset_i)$,
	which is winning a.s. for Player~$\py$ on this domain.
	%$\ywinset(\induced\ylcsstrat_j^2,\xfinitestrats(\game\cut\xset_j\cut\zset_j))(\game\cut\xset_j\cut\zset_j,{\yparity}^{=1})$
	Then, $\ylcsstrat_c$ is defined as
	\begin{align*}
		\ylcsstrat_c := \lcsstrat_1^1 \cup \lcsstrat_1^2 \cup \cdots \cup \lcsstrat_j^1 \cup \lcsstrat_j^2
	\end{align*}
	By reasoning as in the proof of Lemma~\ref{cn:infinite:termination:lemma},
	$\ylcsstrat_c$ induces a correct strategy, i.e.,
	$\gcomplementof\game{\cset_n(\game)}\;\subseteq\;\ywinset(\yinducedlcsstrat_c,\xfinitestrats(\game))(\game,{\yparity}^{>0})$.
\end{proof}

\section{Conclusions and Discussion}
\label{conclusions:section}

We have presented a scheme for solving stochastic games with \as\  and \wpp\  parity winning
conditions under the two requirements that
(i)  the game contains a finite
attractor and 
(ii) both players are restricted to finite-memory strategies.
We have shown that this class of games is memoryless determined.
The method is instantiated to prove decidability
of \as\  and \wpp\  parity games induced by lossy channel systems.

\input{fig_no_finite_attractor}

The two above requirements are both necessary for
our method.
To see why our scheme fails if the game lacks a {\bf finite attractor},
consider the game in Figure~\ref{req:figure}
 (a variant of the Gambler's ruin problem).
All states are random, i.e., $\zstates=\ostates=\emptyset$, 
and $\coloring(\state_0)=1$ and $\coloring(\state_i)=0$ when $i>0$. 
The probability to go right from
any state is $0.7$ and the probability to go left (or to make a
self-loop in $\state_0$) is $0.3$.
This game does not have any finite attractor. 
It can be shown that
the probability to reach $\state_0$ infinitely often is 0 for all initial
states. However, our construction will classify all states as winning
for Player 1.
More precisely, the construction of $\cset_1(\game)$ converges
after one iteration, with $\zset_\alpha=S$ and $\xset_\alpha=\yset_\alpha=\emptyset$ for all $\alpha$, and
$\cset_1(\game)=S$.
Intuitively, the problem is that even if the force-set of $\set{\state_0}$ 
(which is the entire set of states) is visited infinitely many times, 
the probability of visiting $\set{\state_0}$ infinitely often is still zero,
since the probability of returning to $\set{\state_0}$ gets smaller and
smaller. Such behavior is impossible in a game graph that contains
a finite attractor.

Our scheme also fails when the players are not both restricted to {\bf finite-memory strategies}. 
Solving a game under a finite-memory restriction is a
different problem from when arbitrary strategies are allowed (not a sub-problem). 
In fact, it was shown in \cite{BBS:ACM2007} that
for arbitrary strategies, the problem is undecidable.
We show two simple examples of stochastic games on LCSs where the two problems yield different results (see also \cite{BBS:ACM2007}).
In one case, we show that infinite memory is more powerful for Player 1 with a \wpp\  objective
(cf. Figure~\ref{fig:infinite_memory_wpp_example}),
while in the other case infinite memory helps w.r.t. an \as\  objective
(cf. Figure~\ref{fig:infinite_memory_as_example}).
In both cases, Player 0 does not play in the game, thus the memory allowed to her is irrelevant.

\begin{figure}

	\subfloat[W.p.p. winning condition] {%\centering
	    \label{fig:infinite_memory_wpp_example}
		$\qquad$
		\begin{tikzpicture}
			
			\node[opponent-node] (s0) {$p:0$};
			\node[opponent-node] (s1) [right of = s0, node distance=1.8cm] {$q:0$};
			\node[opponent-node] (s2) [right of = s1, node distance=2cm] {$r:1$};

			\draw[transition-edge] (s0) to [out=60, in=120, loop] node [above] {$c!1$} (s0);
			\draw[transition-edge] (s0) to node [above] {$\nop$} (s1);
			\draw[transition-edge] (s1) to [out=60, in=120, loop] node [above] {$\nop$} (s1);
			\draw[transition-edge] (s1) to node [above] {$c?1$} (s2);
			\draw[transition-edge] (s2) to [bend left, out=30, in=150] node [below] {$\nop$} (s0);
	
		\end{tikzpicture}
		$\qquad$
	}
	$\quad$
	\subfloat[A.s. winning condition] {
		\label{fig:infinite_memory_as_example}
		$\qquad$
		\begin{tikzpicture}
			
			\node[opponent-node] (s0) {$0$};
			\node[opponent-node] (s1) [right of = s0, node distance=1.6cm] {$1$};
			\node[opponent-node] (s2) [right of = s1, node distance=1.6cm] {$2$};

			\draw[transition-edge] (s0) to [out=60, in=120, loop] node [above] {$c!1$} (s0);
			\draw[transition-edge] (s0) to [bend left] node [above] {$\nop$} (s1);
			\draw[transition-edge] (s1) to [bend left] node [below] {$c?1$} (s0);
			\draw[transition-edge] (s1) to node [above] {$\nop$} (s2);
			\draw[transition-edge] (s2) to [bend left, out=90, in=90] node [below] {$\nop$} (s0);
	
		\end{tikzpicture}
		$\qquad$
	}
	\caption{Infinite memory helps Player 1.}
\end{figure}

First, we show that infinite memory is more powerful for \wpp\  objectives.
In Figure~\ref{fig:infinite_memory_wpp_example}, Player~1 plays on control states $p$, $q$, and $r$.
Player 1's objective is to visit state $r$ infinitely often w.p.p..
To ensure this, from state $p$ Player~1 pumps up the channel to a sufficiently large size $k$
(which can be done a.s. for any $k$ given enough time),
and then she goes to the risk state $q$. If each message can be lost independently with probability $\frac 1 2$,
the probability that all messages are lost, and thus that Player 1 is stuck forever in $q$, is $2^{-k}$.
Otherwise, with probability $1 - 2^{-k}$ Player 1 can visit $r$ once,
and then go back to $p$.
The strategy of Player 1 is to realise an infinite sequence $k_0 < k_1 < \cdots$
s.t. the probability of visiting state $r$ infinitely often,
which is $\prod_{i = 0}^\infty (1 - 2^{-k_i})$,
can be made strictly positive.
Clearly, if Player 1 has infinite memory, then she can realize such a sequence by distinguishing different visits to control state $p$ and same channel contents.
On the other side, if Player 1 is restricted to finite memory, then either the game eventually stays forever in $p$ (which is losing),
or the infinite sequence $k_0, k_1, \dots$ is upper-bounded by some finite $n$,
which makes the infinite product above equal to $0$.
In both cases, Player 1 loses if she has only finite memory.

Notice that Player~1 wins not only w.p.p., but even limit-sure in this example.
In other words, for every $\varepsilon > 0$ there is an infinite-memory strategy
s.t. the parity objective is satisfied with probability $\geq \varepsilon$.
We don't know whether there are examples where a similar phenomenon can be reproduced under finite-memory/memoryless strategies.

We now show that infinite memory is more powerful for \as\ objectives.
An example similar to the previous case can be given for the a.s. winning mode with a 3-color parity condition.
In Figure~\ref{fig:infinite_memory_as_example}, Player 1 controls states $0$, $1$, and $2$,
whose color equals their name.
Thus, the objective of Player 1 is to a.s. visit state $1$ infinitely often and state $2$ only finitely often.
The strategy is similar as in the previous example:
Player 1 tries to pump up the channel in state $0$,
and then she goes to the risk state $1$.
From here, with low probability all messages are lost, and the penalty is to visit state $2$ once.
Otherwise, the game can go back directly to state $0$ without visiting state $2$. 
In both cases, the game restarts afresh from state $0$.
An analysis as in the previous example shows that, if Player 1 is restricted to finite memory,
then the probability of visiting state $2$ from state $1$ can be bounded from below.
This implies that, whenever state $1$ is visited infinitely often, then so is state $2$ a.s.,
and so Player 1 is losing.
On the other hand, there is an infinite-memory strategy for Player 1
s.t. the probability of visiting state $2$ for $n$ times goes to $0$ as $n$ goes to infinity,
which implies that the probability of visiting state $2$ only finitely often is 1.

As future work, we will consider extending our framework
to (fragments of) probabilistic extensions of other models
such as Petri nets and noisy Turing machines
\cite{Parosh:etal:MC:infinite:journal}.

% \section*{Acknowledgments}
% We thank the anonymous referees for their detailed comments.

\bibliographystyle{alpha}
\bibliography{bibdatabase,private}

\end{document}

%% file: fig_parity_schema.tex
\begin{figure}
	\centering
	\begin{tikzpicture}[spy using overlays={size=12mm}]
		
		\begin{scope}
			\node at (-1,-1) {game $\game$};
			\draw [clip] (0,0) rectangle (6.2,-2);
			\fill [fill = gray!20] (4,0) rectangle (6.2,-2);
			\node at (5.2,-1) {$\bigcup_{\beta < \alpha} \yset_\beta$};
			\begin{scope}
				\path [clip] (0,0) rectangle (4,-2);
				\node [starburst, draw, minimum height = 4cm, minimum width = 3.6cm, fill = gray!20] at (3.5,-1) {};
			\end{scope}
			\draw[dashed,very thick,->] (3.5,-1) -- node [above] {$1-x$} (4.4,-1);
			\draw (4,0) -- (4,-2);
			\node at (3,-1) {$\xset_\alpha$};
			\node at (1,-1) {$S \setminus \xset_\alpha$};
		\end{scope}
		
		\draw (0, -2) -- (1, -3);
		\draw (2, -2) -- (7.5, -3);
		
		\begin{scope}
			\node at (-.4,-4) {game $\game\cut\xset_\alpha$};
			\draw [clip] (1,-3) rectangle (7.5,-5);
			%\draw [color = gray!100] (2,-3) -- (2,-5);
			%\fill [fill = gray!20, draw = none] (2,-3) rectangle (4,-5);
			\begin{scope}
				\path [clip] (1,-3) rectangle (5,-5);
				\node [starburst, draw, minimum height = 4cm, minimum width = 3.6cm, fill = white] at (4.5,-4) {};
			\end{scope}
			\node at (6.4,-4) {$\colorset{\gcomplementof\game\xset_\alpha}=n$};
			\draw [dashed, very thick,->] (4.5,-4) -- node [above] {$x$} (5.4,-4);
			\draw (5,-3) -- (5,-5);
			\node at (4,-4) {$\zset_\alpha$};
			\node at (2,-4) {$S \setminus \xset_\alpha \setminus \zset_\alpha$};
		\end{scope}
		
		\draw (1, -5) -- (-1, -6);
		\draw (3, -5) -- (8, -6);
		
		\begin{scope}
			\node at (-3,-6.75) {game $\game\cut\xset_\alpha\cut\zset_\alpha$};
			\draw (-1,-6) rectangle (8,-7.5);
			\node at (1.4,-6.75) {$\gcomplementof{\game\cut\xset_\alpha\cut\zset_\alpha}\cset_{n-1}(\game\cut\xset_\alpha\cut\zset_\alpha)$};
			\fill [fill = gray!20] (4,-6) rectangle (8,-7.5);
			\draw (4,-6) -- (4,-7.5);
			\node at (6,-6.75) {$\cset_{n-1}(\game\cut\xset_\alpha\cut\zset_\alpha)$};
		\end{scope}
		
	\end{tikzpicture}
	\caption{The construction of the various sets involved in the inductive step.
	  The grey area is $\yset_\alpha$.}
	\label{fig:parity:schema}
\end{figure}

%% file: fig_no_finite_attractor.tex
\begin{figure}
\begin{tikzpicture}[show background rectangle]
\node[name=dummy] {};
\node[prob-node,name=s0] at (dummy) {$s_0$};
\node[prob-node,name=s1,anchor=west] at ($(s0.east)+(3mm,0mm)$) {$s_1$};
\node[prob-node,name=s2,anchor=west] at ($(s1.east)+(3mm,0mm)$) {$s_2$};
\node[prob-node,name=s3,anchor=west] at ($(s2.east)+(3mm,0mm)$) {$s_3$};
\node[name=s4,anchor=west] at ($(s3.east)+(2mm,0mm)$) {$\cdots$};

%\node[anchor=south] at ($(s4.north)+(-5mm,5mm)$) {(a)};

\draw[transition-edge] (s0) to [out=150,in=210,loop] node[above,near start]{\scriptsize $0.3$} (s0);
 
\draw[transition-edge] (s0) to [out=60,in=120] node[above]{\scriptsize $0.7$} (s1);
\draw[transition-edge] (s1) to [out=240,in=300] node[below]{\scriptsize $0.3$} (s0);

\draw[transition-edge] (s1) to [out=60,in=120] node[above]{\scriptsize $0.7$} (s2);
\draw[transition-edge] (s2) to [out=240,in=300] node[below]{\scriptsize $0.3$} (s1);

\draw[transition-edge] (s2) to [out=60,in=120] node[above]{\scriptsize $0.7$} (s3);
\draw[transition-edge] (s3) to [out=240,in=300] node[below]{\scriptsize $0.3$} (s2);

%\draw[dotted,line width=1pt] (s3) -- (s4);
\end{tikzpicture}

\caption{ Finite attractor requirement. }
\label{req:figure}
\end{figure}